\documentclass[11pt]{article}

\usepackage{amsthm}
\usepackage{graphicx} 
\usepackage{array} 

\usepackage{amsmath, amssymb, amsfonts, verbatim}
\usepackage{hyphenat,epsfig,subfigure,multirow}

\usepackage[usenames,dvipsnames]{xcolor}
\usepackage[ruled]{algorithm2e}

\newcommand {\ignore} [1] {}

\usepackage{tcolorbox}
\tcbuselibrary{skins,breakable}
\tcbset{enhanced jigsaw}

\usepackage[compact]{titlesec}

\definecolor{DarkRed}{rgb}{0.5,0.1,0.1}
\definecolor{DarkBlue}{rgb}{0.1,0.1,0.5}

\usepackage{nameref}
\definecolor{ForestGreen}{rgb}{0.1333,0.5451,0.1333}
\definecolor{Red}{rgb}{0.9,0,0}
\usepackage[linktocpage=true,
	pagebackref=true,colorlinks,
	linkcolor=DarkRed,citecolor=ForestGreen,
	bookmarks,bookmarksopen,bookmarksnumbered]
	{hyperref}
\usepackage[noabbrev,nameinlink]{cleveref}

\crefname{property}{property}{Property}
\creflabelformat{property}{(#1)#2#3}
\crefname{equation}{eq}{Eq}
\creflabelformat{equation}{(#1)#2#3}

\usepackage{bm}
\usepackage{url}
\usepackage{xspace}
\usepackage[mathscr]{euscript}

\usepackage{tikz}
\usetikzlibrary{arrows}
\usetikzlibrary{arrows.meta}
\usetikzlibrary{shapes}
\usetikzlibrary{backgrounds}
\usetikzlibrary{positioning}
\usetikzlibrary{decorations.markings}
\usetikzlibrary{patterns}
\usetikzlibrary{calc}
\usetikzlibrary{fit}
\usetikzlibrary{snakes}

\usepackage{mdframed}

\usepackage[noend]{algpseudocode}
\makeatletter
\def\BState{\State\hskip-\ALG@thistlm}
\makeatother

\usepackage{cite}
\usepackage{enumitem}

\def\MathE{\hbox{{\rm I}\hskip -2pt {\rm E}}}
\newcommand{\Expect}{\MathE}

\usepackage[margin=1in]{geometry}

\newtheorem{theorem}{Theorem}
\newtheorem{lemma}{Lemma}[section]

\newtheorem{invariant}[lemma]{Invariant}
\newtheorem{assumption}{Assumption}
\newtheorem{corollary}[lemma]{Corollary}
\newtheorem{claim}[lemma]{Claim}
\newtheorem{fact}[lemma]{Fact}

\newtheorem{definition}[lemma]{Definition}

\crefname{lemma}{Lemma}{Lemmas}
\crefname{invariant}{Invariant}{Invariants}
\crefname{claim}{Claim}{Claims}
  
\newtheorem*{claim*}{Claim}
\newtheorem*{proposition*}{Proposition}
\newtheorem*{lemma*}{Lemma}
\newtheorem*{problem*}{Problem}
\newtheorem*{theorem*}{Theorem}

\newtheorem{mdresult}{Result}

\theoremstyle{definition}
\newtheorem{remark}[lemma]{Remark}
\newtheorem{observation}[lemma]{Observation}

\allowdisplaybreaks

\DeclareMathOperator*{\argmax}{arg\,max}

\renewcommand{\qed}{\nobreak \ifvmode \relax \else
      \ifdim\lastskip<1.5em \hskip-\lastskip
      \hskip1.5em plus0em minus0.5em \fi \nobreak
      \vrule height0.75em width0.5em depth0.25em\fi}

\newcommand{\eps}{\ensuremath{\varepsilon}}

\newcommand{\paren}[1]{\ensuremath{\left(#1\right)}\xspace}
\newcommand{\card}[1]{\left\vert{#1}\right\vert}

\newcommand{\set}[1]{\ensuremath{\left\{ #1 \right\}}}
\newcommand{\poly}{\mbox{\rm poly}}

\DeclareMathOperator*{\Prob}{\ensuremath{\textnormal{Pr}}}

\newcommand{\etal}{{\it et al.\,}}

\newenvironment{tbox}{\begin{tcolorbox}[
		enlarge top by=5pt,
		enlarge bottom by=5pt,
		 breakable,
		 boxsep=0pt,
                  left=4pt,
                  right=4pt,
                  top=10pt,
                  arc=0pt,
                  boxrule=1pt,toprule=1pt,
                  colback=white
                  ]
	}
{\end{tcolorbox}}

\newcommand{\II}{\ensuremath{\mathbb{I}}}
\newcommand{\HH}{\ensuremath{\mathbb{H}}}

\newcommand{\estar}{\ensuremath{e^{\star}}}

\def\cE{{\cal E}}

\title{Fully Dynamic Set Cover via Hypergraph Maximal Matching: \\An Optimal Approximation Through a Local Approach}
\author{Sepehr Assadi \and Shay Solomon}

\date{}

\begin{document}
\maketitle

\pagenumbering{roman}


\begin{abstract}

In the (fully) dynamic set cover problem, we have a collection of $m$ sets from a universe of size $n$ that undergo element insertions and deletions; the goal is to maintain an approximate set cover of the universe after each update. 
We give an $O(f^2)$ update time algorithm for this problem that achieves an $f$-approximation, where $f$ is the maximum number of sets that an element belongs to; under the unique games conjecture, this 
approximation is best possible for any fixed $f$.  This is the first algorithm for dynamic set cover with
approximation ratio that \emph{exactly} matches $f$ (as opposed to \emph{almost} $f$ in prior work), as well as the first one with runtime \emph{independent of $n,m$} (for any approximation factor of $o(f^3)$). 

\medskip

Prior to our work, the state-of-the-art algorithms for this problem were $O(f^2)$ update time algorithms of Gupta \etal [STOC'17] and Bhattacharya \etal [IPCO'17] with $O(f^3)$ approximation, 
and the  recent algorithm of Bhattacharya~\etal~[FOCS'19] with $O(f \cdot \log{n}/\eps^2)$ update time and $(1+\eps) \cdot f$ approximation, improving the $O(f^2 \cdot \log{n}/\eps^5)$ bound of Abboud~\etal~[STOC'19]. 

\medskip

The key technical ingredient of our work is an algorithm for maintaining a \emph{maximal} matching in a dynamic hypergraph of rank $r$---where each hyperedge has at most $r$ vertices---that undergoes hyperedge insertions and deletions 
in $O(r^2)$ amortized update time; our algorithm is randomized, and the bound on the update time holds in expectation and with high probability. This result generalizes the maximal matching algorithm of Solomon [FOCS'16] with constant update time in ordinary graphs to hypergraphs, and is of independent merit; the previous state-of-the-art algorithms for set cover do not translate to (integral) matchings for hypergraphs, let alone a maximal one.
Our quantitative result for the set cover problem is translated directly from this qualitative result for maximal matching
using standard reductions. 
 
 \medskip
 
An important advantage of our approach over the previous ones for approximation $(1+\eps) \cdot f$ (by Abboud~\etal~[STOC'19] and Bhattacharya~\etal~[FOCS'19]) is that it is inherently {\em local} and can thus be distributed efficiently to achieve low amortized round and message complexities.

\end{abstract}

\clearpage

\setcounter{tocdepth}{2}
\tableofcontents

\clearpage

\pagenumbering{arabic}
\setcounter{page}{1}

\newcommand{\DD}{\ensuremath{\mathcal{D}}}
\newcommand{\addrandom}{\ensuremath{\textnormal{\textsf{insert-hyperedge}}\xspace}}
\newcommand{\handleedge}{\ensuremath{\textnormal{\textsf{handle-hyperedge}}\xspace}}
\newcommand{\FF}{\mathcal{F}}
\renewcommand{\SS}{\mathcal{S}}

\section{Introduction}\label{sec:intro}

In the set cover problem, we are given a family $\SS = (S_1,\ldots,S_m)$ of $m$ sets on a universe $[n]$. The goal is to find a minimum-size subfamily of sets $\FF \subseteq \SS$ such that 
$\FF$ covers all the elements of $[n]$. Throughout the paper, we use $f$ to denote the maximum {frequency} of any element $i \in [n]$ inside the sets in $\SS$ (by frequency of element $i$, we mean the number of sets
in $\SS$ that contain $i$). 

The set cover problem is one of the most fundamental and well-studied NP-hard problems, with two classic approximation algorithms, both with runtime $O(f n)$: greedy $\ln n$-approximation and primal-dual $f$-approximation.
One cannot achieve approximation $(1-\eps) \ln n$ unless P = NP \cite{williamson2011design,dinur2014analytical}
or approximation $f-\eps$ for any fixed $f$ under the unique games conjecture \cite{KR08}.

In recent years there is a growing body of work on this problem in the dynamic setting, where one would like to efficiently maintain a set cover for a universe that is subject to element insertions and deletions.
The holy grail is to coincide with the bounds of the static setting: 
approximation factor of either $\ln{n}$ or $f$ with (amortized) update time $O(f)$ (as a static runtime of $\Theta(f n)$ means that we spend $\Theta(f)$ time per each element of the universe on average). 
Indeed, in any dynamic model where element updates are explicit, update time $\Omega(f)$ is inevitable.
Even in stronger models where updates are supported implicitly in constant time, 
recent SETH-based conditional lower bounds 
imply that update time $O(f^{1-\eps})$ requires polynomial approximation factor \cite{abboud2019dynamic}.

The dynamic set cover problem was first studied by Gupta~\etal~\cite{GKKP17} who gave an $O(\log n)$-approximation with update time $O(f \log n)$ based on a greedy algorithm. 
The rest of the known algorithms are primal-dual-based and their approximation factor depends only on $f$.
The state-of-the-art algorithms  are: 
\begin{itemize}
\item An $O(f^3)$-approximation with
$O(f^2)$ update time, by Gupta \etal \cite{GKKP17} and independently by Bhattacharya \etal \cite{BCH17}; 
\item A $(1+\eps) \cdot f$-approximation with update time $O(f \cdot \log{n}/\eps^2)$
by Bhattacharya~\etal \cite{bhattacharya2019new},  
improving the $O(f^2 \cdot \log{n}/\eps^5)$ bound of Abboud~\etal \cite{abboud2019dynamic} 
as well as an earlier result by Bhattacharya \etal \cite{BhattacharyaHI15}. 
\end{itemize}
\noindent
Interestingly, the state-of-the-art algorithms for this problem are all {\em deterministic} 
(the algorithm of~\cite{abboud2019dynamic} is  randomized however).

\paragraph{Our Result.} In this work we demonstrate the power of randomization for the dynamic set cover problem by achieving the best possible approximation of $f$ with runtime independent of both $m,n$: 

\begin{theorem*}[\textbf{Informal}]
	There is an algorithm for the dynamic set cover problem that achieves an \emph{exact} $f$-approximation in $O(f^2)$ expected (and with high probability) amortized update time. 
\end{theorem*}

This gives the first algorithm for dynamic set cover with approximation ratio that \emph{exactly} matches $f$ (as opposed to \emph{almost} $f$ in prior work), as well as the first one with runtime \emph{independent of $n,m$} (for any approximation factor of $o(f^3)$). The bound $O(f^2)$ on the update time of our algorithm holds with high probability and in expectation. As in \cite{abboud2019dynamic} and in the great majority of randomized dynamic graph algorithms, we assume an oblivious adversary. 
We shall remark that even for the much simpler problem of dynamic vertex cover (corresponding to $f=2$ case), no algorithm is known against an adaptive adversary that achieves an exact $2$-approximation in time independent of other input parameters (although $(2+\eps)$-approximation algorithms have been known for some time now~\cite{BhattacharyaK19,BCH17,BHI15,BHN17}).

\paragraph{Maximal Hypergraph Matching.} The key technical ingredient of our work is an algorithm for maintaining a \emph{maximal} matching in a dynamic hypergraph of rank $r$---where each hyperedge has at most $r$ vertices---that undergoes hyperedge insertions and deletions 
in $O(r^2)$ amortized update time.
This result generalizes the maximal matching algorithm of Solomon \cite{Sol16} with constant update time for ordinary graphs, and is of independent merit; the previous state-of-the-art algorithms for set cover do not translate to (integral) matchings for hypergraphs, let alone a maximal one (however the algorithm of~\cite{abboud2019dynamic} translates to
an integral (non-maximal) matching). The result for set cover follows immediately from this: taking all the matched vertices in a maximal hypergraph matching of rank $f$ yields an $f$-approximate hypergraph vertex cover, or equivalently, an $f$-approximate set cover; 
(see~\Cref{sec:prelim} for details of this standard reduction).

We stress that there is an inherent difference between a maximal matching (yielding exactly $f$-approximation) and an almost-maximal matching
(yielding $(1+\eps) \cdot f$-approximation), wherein an $\eps$-fraction of the potentially matched vertices may be unmatched; this is true in general but particularly important in the dynamic setting. 
Despite extensive work on dynamic algorithms for matching and vertex cover in ordinary graphs ($f = 2$), the state-of-the-art deterministic algorithm (or even randomized against adaptive adversary) for 2-approximate matching and vertex cover (via a maximal matching) has update time $O(\sqrt{|E|})$ \cite{NS13}, while $(2+\eps)$-approximate matching and vertex cover can be maintained deterministically in poly-log update time \cite{BHI15,BhattacharyaHN16}; for approximate vertex cover and \emph{fractional} matching, the update time can be further reduced to $O(1/\eps^2)$ \cite{BhattacharyaK19}. 
Therefore, it is only natural that our algorithm, which achieves an integral maximal hypergraph matching, is randomized and assumes an oblivious adversary.

\paragraph{Distributed networks.}
There is a growing body of work on distributed networks that change dynamically 
(cf.\ \cite{ParterPS16,PS16,CHK16,AOSS18,KS18,BC19,DBLP:journals/corr/abs-2010-16177}).
A distributed network can be modeled as an undirected (hyper)graph $G(V,E)$, where each vertex $v \in V$ is identified with a unique processor and the edge set $E$ corresponds to direct communication links between the processors. 
In a static distributed setting all processors wake up simultaneously,
and computation proceeds in fault-free synchronous rounds during which every processor exchanges messages of size $O(\log n)$ with its neighbors.\footnote{We consider the standard $\mathcal{CONGEST}$ model (cf.~\cite{PelB00}), which captures the essence of spatial locality and congestion. For hypergraphs, the messages should be of size $O(\log m)$, where $m$ is the number of edges in the graph.}

We focus on the standard setting in dynamic graph algorithms---dynamically changing networks that undergo edge updates (both insertions and deletions, a single edge update per step), which initially contain no edges. 
But in a {\em distributed} network we are subject to the following {\em local} constraint: After each edge update, only the affected vertices---the endpoints of the updated edge---are woken up; this is referred to in previous work as the {($\mathcal{CONGEST}$) \em local wakeup model}. 
Those affected vertices initiate an update procedure, which also involves waking up the vertices in the network (beyond the affected ones) that are required to participate in the update procedure, to adjust all outputs to agree on a valid global solution---in our case a maximal hypergraph matching; the output of each vertex is the set of its incident edges that belong to the matching. 
We make a standard assumption that the edge updates occur in large enough time gaps, and hence the network is always ``stable'' before the next change occurs (see, e.g.,~\cite{ParterPS16,CHK16,AOSS18}). 
In this setting, the goal is to optimize (1) the number of communication rounds, and (2) the number of messages,
needed for repairing the solution per edge update,  over a worst-case sequence of edge updates.

Our dynamic maximal hypergraph matching algorithm can be naturally adapted to \emph{distributed networks}.
Note that following an edge update, $O(1)$ communication rounds trivially suffice for updating a maximal hypergraph matching. 
However, the number of messages sent per update via this naive algorithm may be a factor of $r^2$ greater than the maximum degree in the hypergraph, which could be $\Omega({n \choose r-1})$, where $n$ is the number of vertices and $r$ is the rank.
An important objective is to design a dynamic distributed algorithm that achieves, in addition to low round complexity, a low {\em message complexity}, which is obtained by our algorithm as mentioned below. 

The inclusion-maximality of our maintained matching enforces our algorithm to work persistently so there is never any ``slack''; that is, the algorithm makes sure that every edge that can be added to the matching is added to it, which stands in contrast to ``lazy'' approximate-maximal matching (or approximate set cover) algorithms, which may wait to accumulate an $\eps$-factor additive slack in size or weight, and only then run an update procedure. 
However, such a lazy update procedure is inherently non-local, where, following an edge update $e$, the required changes to the maintained graph structure may involve edges and vertices that are arbitrarily far from $e$; indeed, this is the case with the previous algorithms that achieve approximation factor close to $f$ \cite{bhattacharya2019new,abboud2019dynamic}; moreover, the ``lazy'' feature of these algorithms must rely on a centralized agent that orchestrates the update algorithm with the use of global data structures,  and this is, in fact, the key behind the efficiency of these algorithms  \cite{bhattacharya2019new,abboud2019dynamic}.  
Our algorithm, on the other hand, does not employ any global update procedure or global data structures, and as such it is inherently {\em local} and can be easily distributed, so that the average number of messages sent per update is $O(r^2)$, matching the sequential update time. The number of rounds is clearly upper bounded by the number of messages.

Some more details on this application of our main result are deferred to Section \ref{appdist}.

\paragraph{Recent related work:}
Independently and concurrently to our work, Bhattacharya, Henzinger, Nanongkai, and Wu~\cite{DBLP:journals/corr/abs-2002-11171,BhattacharyaHNW21} obtained an algorithm for dynamic set cover  with 
	$O(f^2/\eps^3)$ amortized update time and $O(f \cdot \log^2{n}/\eps^3)$ worst case update time and $(1+\eps) \cdot f$ approximation. While closely related, their results and ours are incomparable. Our algorithm can achieve a better approximation ratio of \emph{exactly} $f$ as opposed to $(1+\eps) f$ (which is the 
	first dynamic algorithm with this guarantee) but is randomized, while their result is deterministic and can work for weighted set cover (with extra $\log{C}$-dependence on the update time where $C$ is the maximum weight of any set). 
		Also, as mentioned, our algorithm is inherently local and can be distributed efficiently, whereas the algorithm of Bhattacharya~\etal, as the previous aforementioned algorithms, is non-local. 
	In terms of techniques, the two works are  entirely disjoint.  

\subsection{Technical overview}

At a  high level, all the previous state-of-the-art algorithms
 \cite{GKKP17,BCH17,bhattacharya2019new} (as well as the independent work of~\cite{DBLP:journals/corr/abs-2002-11171,BhattacharyaHNW21})  follow a deterministic primal-dual approach by maintaining a fractional packing solution as a dual certificate.
The main advantage of this approach over ours, of course, is in being deterministic, and its drawbacks are that 
(1) the approximation factor is almost $f$ rather than exactly $f$, (2) it does not give rise to an \emph{integral} matching in the context of hypergraphs, and (3) it is non-local.

Our algorithm generalizes and strengthens the maximal matching algorithm by Solomon \cite{Sol16}, which, in turn, builds on and refines the pioneering approach of Baswana \etal \cite{BGS11}. For conciseness, we shall sometimes refer only to \cite{Sol16} in the following (to avoid explaining the differences between \cite{Sol16} and \cite{BGS11}); of course, by that we do not mean to take any credit of \cite{BGS11}, on which \cite{Sol16} relies. 

\paragraph{Maximal matching algorithm of~\cite{BGS11,Sol16}.} Consider a deletion of a matched edge $(u,v)$ from the graph (for $r=2$), which is the only nontrivial update. 
Focusing on $u$, if $u$ has an unmatched neighbor, we need to match $u$.  
To avoid a naive scan of the neighbors of $u$ (requiring $O(n)$ time),
the key idea is to match $u$ with a randomly sampled (possibly matched) neighbor $w$.
Under the oblivious adversarial model, the expected number of edges incident on $u$ that are deleted from the graph before the deletion of edge $(u,w)$ is roughly half the  ``sample space'' size of $u$, i.e., the number of neighbors of $u$ from which we sampled $w$, which can be viewed as ``time token'' (or potential value) that is expected to arrive in the future. 

To benefit from these tokens,~\cite{BGS11,Sol16} introduces a \emph{leveling scheme} where the \emph{levels} of vertices are exponentially smaller estimates of these potential values: Unmatched vertices have level $-1$ and 
matched vertices are assigned the levels of their matched edge, which is roughly the logarithm of the sample space size.
This defines a dynamic hierarchical partition of \emph{vertices} into $O(\log n)$ levels.

A key ingredient in the algorithm of \cite{BGS11, Sol16} is the {\em sampling rule}: A random mate $w$ is chosen for a vertex $u$ among its neighbors of strictly lower level. Intuitively, if $w$'s level is lower than $u$'s, then $w$'s potential value is much smaller than $u$'s, and the newly created matched edge $(u,w)$ provides enough potential  to cover the cost of deleting the old matched edge on $w$ (if any).

\paragraph{Difficulties of going from ordinary graphs to hypergraphs.} The main difficulty of extending the prior work in~\cite{BGS11,Sol16} 
to hypergraphs of rank $r > 2$ has to do with the ``vertex-centric'' approach taken by these works~\cite{BGS11,Sol16}. For instance, even at the definition level, 
it is already unclear how to generalize the sampling rule of the algorithm for hypergraphs, even for $r = 3$,
since the hyperedges incident on $u$ may consist of endpoints of various levels, some smaller than that of $u$, some higher.
Ideally we would want to sample the matched hyperedge among those where {\em all} endpoints have lower level than that of $u$, but it is a-priori unclear how to maintain this hyperedge set efficiently. 
In particular, to perform the sampling rule efficiently, the strategy of \cite{Sol16} is to dynamically orient each edge towards the lower level endpoint, and a central obstacle is to efficiently cope with edges where both endpoints are at the same level. For hypergraphs, naturally, these obstacles become more intricate considering there are $r$ different endpoints now.

There are also other hurdles that need to be carefully dealt with, such as the following.
Say a matched hyperedge $e = (u_1,\ldots,u_r)$ gets deleted from the hypergraph. 
When $r = 2$, any edge incident on $u_1$ is different than any edge incident on $u_2$, hence informally the update algorithm can handle $u_1$ and $u_2$ {\em independently} of each other. 
When $r >2$ (even for $r = 3$), different endpoints of $e$ may share (many) hyperedges in common. 
Therefore, when choosing random matched hyperedges for the newly unmatched endpoints of $e$, we need to (i) be careful not to create conflicting matched hyperedges, but at the same time (ii) keep the sample spaces of endpoints
sufficiently large so that the potential values can cover the runtime of the update procedure;  balancing between these two contradictory requirements is a key challenge in our algorithm. 

\paragraph{An $O(r^3)$ update time algorithm.} We manage to cope with these and other hurdles by instead switching to a ``hyperedge-centric'' view of the algorithm. Informally speaking, this means that 
instead of letting vertices derive the potential values and levels, we assign these values to the hyperedges and use those to define the corresponding level for remaining vertices. 
Under this new view of the algorithm, we can indeed generalize the approach of~\cite{Sol16} to obtain an $O(r^4)$-update time algorithm for rank $r$ hypergraphs. We note that considering the intricacies in~\cite{Sol16}, 
already achieving this $O(r^4)$ time bound turned out to be considerably challenging.

Improving the update time to $O(r^3)$ is based on the following insight. In this hyperedge-centric view, 
a level-$\ell$ matched hyperedge $e$ is sampled to the matching from a sample space $S(e)$ of size roughly $\alpha^{\ell}$, where $\alpha = \Theta(r)$. We refer to the hyperedges in $S(e)$ as the {\em core hyperedges} of $e$. 
All the core hyperedges of $e$ are then also assigned a level $\ell$. Let us focus on an $e' = (u_1,\ldots,u_r) \in S(e)$. 
Subsequently, $u_1$ may initiate the creation of a new matched hyperedge at level $\ell' > \ell$, 
at which stage we randomly sample a new matched hyperedge, denoted by $e_1$, 
among all its incident hyperedges of level lower than $\ell'$, so $e' \in S(e_1)$, i.e., hyperedge $e'$ is now a core hyperedge of $e_1$ as well. 
Perhaps later $u_2$ may initiate the creation of a new matched hyperedge $e_2$ at level $\ell'' > \ell'$, and then hyperedge $e'$ will
become a core hyperedge of $e_2$, and so on and so forth. Thus any hyperedge may serve as a core hyperedge of up to $r$ matched hyperedges at any point in time. 

To shave a factor of $r$ from the runtime, we need to make sure that each hyperedge serves as a core hyperedge of a \emph{single} matched hyperedge.
This is achieved by ``freezing''  (or \emph{temporarily} deleting) all core hyperedges of a newly created matched hyperedge $e$; in the sequel (see Section~\Cref{sec:deleted}) we shall refer to these hyperedges as {\em temporarily deleted hyperedges} (due to $e$) rather than ``core hyperedges'', and they will comprise the set  $\DD(e)$.
Then, whenever the matched hyperedge gets deleted from the matching, 
we need to ``unfreeze'' these core hyperedges and update all their \emph{ignored} data structures -- our analysis shows that this can be carried out efficiently. 

Since this algorithm already requires an entirely new view of the previous approaches for ordinary graphs and several nontrivial ideas, we present it as a standalone result 
in~\Cref{sec:r4}. 

\paragraph{An $O(r^2)$ update time algorithm.} The main step however is to improve the update time from $O(r^3)$ to $O(r^2)$. We next provide a couple of technical highlights behind this improvement.

In our algorithm, the potential values of level-$\ell$ matched hyperedges are in the range $[\alpha^{\ell},\alpha^{\ell+1})$, for $\alpha = O(r)$,
hence they may differ by a factor of $O(r)$.
Consider the moment a level-$\ell$ matched hyperedge $e = (u_1,\ldots,u_r)$ gets deleted by the adversary; the leveling scheme allows us to assume 
we have an $O(\alpha^{\ell+2})$ ``potential time'' for handling this hyperedge. To get an update time of $O(r^2)$, we need to handle each of the endpoints $u_i \in e$ within time $O(\alpha^{\ell+1})$ 
time or instead ``contribute'' to the potential by creating a new matched hyperedge. If $u_i$ has more than $\alpha^{\ell+1}$ incident hyperedges of level at most $\ell$,
we can sample a random hyperedge among them to be added to the matching, thereby creating a level-$(\ell+1)$ matched hyperedge;
we discuss some issues related to the creation of matched edges later. 
But if $u_i$ has slightly less than $\alpha^{\ell+1}$ incident edges,
this sample space size suffices only for creating a level-$\ell$ matched edge, but it is crucial that the sample space of a level-$\ell$ matched edge would consist only of edges where all endpoints have level strictly lower than $\ell$. 
Even checking whether this is the case is too costly, since iterating over all endpoints of all such edges takes time (slightly less than) $O(\alpha^{\ell+2})$,
and if we are in the same scenario for each $u_i$ this gives rise to a runtime of $O(\alpha^{\ell+3})$,
and thus to an amortized update time of $O(r^3)$.
Even if we could check this for free, if most of the edges have endpoints of level $\ell$, we need to find those endpoints, and to pass the ``ownership'' of the edges to those endpoints. 
To cope with these issues, we maintain a data structure for each hyperedge $e$  that keeps track of all its endpoints of highest level and in $O(1)$ time returns an arbitrary such endpoint or reports that none exists.
Of course, now the challenge becomes maintaining these data structures for the edges with $O(r^2)$ update time.

Consider the moment that a level-$\ell$ matched edge $e$ is created. 
At this stage, another crucial invariant tells us to ``raise'' all endpoints of edge $e$ to level $\ell$, and then to update the ownership set of each endpoint according to its up-to-date level. 
One challenging case is when each endpoint $u_i$ of $e$ now owns slightly less than $\alpha^{\ell+1}$ new edges.
This sample space size does not suffice for creating a level-$(\ell+1)$ matched edge yet it is too costly to update each of the endpoints of these edges about the up-to-date level of $u_i$; summing over all endpoints of $e$, this again gives rise to 
update time of $O(r^3)$. However, if we are equipped with the aforementioned data structure, we can efficiently focus on the edges where all endpoints have level lower than $\ell$, and can thus create a level-$\ell$ matched edge. This is not enough, however, since we are merely replacing one level-$\ell$ matched edge by another, and this process could repeat over and over.
Our goal would be to replace one level-$\ell$ matched edge by \emph{at least two others}, so as to provide enough increase in ``potential'' to cover for this runtime.
Since the total number of edges incident on $e$ in this case is slightly less than $\alpha^{\ell+2}$, the intuition is that we should be able to easily achieve here a fan-out of 2, and therefore a valid charging argument.
Alas, there is one significant caveat when working with hypergraphs, which we already mentioned above--- dependencies.
It is possible that the first level-$\ell$ matched edge that we randomly sampled for $u_1$ intersects all the edges incident on $e$, which will result in fan-out 1.

To overcome this final obstacle, we take our hyperedge-centric view to the next level by employing a new sampling method different than~\cite{BGS11,Sol16} altogether. 
In particular, in this case, our sample space is not necessarily restricted to hyperedges incident on a single vertex, but could be an arbitrary hyperedge set as a function of the deleted hyperedge; however, importantly, to achieve a {\em local} sampling method, we will make sure that the entire sample space is incident to the endpoints of a single edge.
This new sampling method (see Procedure $\addrandom$ in~\Cref{sec:update-r2}) entails a few technical complications primarily to ensure that we still get enough ``potential'' from the adversary, but it ultimately enables us to achieve the desired update time bound of $O(r^2)$. 

\paragraph{The role of randomness in our algorithm.} 
Our algorithm relies crucially on randomization and on the oblivious adversary assumption, and indeed the probabilistic analysis employed in this work is highly nontrivial. In particular, the usage of randomization for reducing the update time bound from $O(f^3)$ to $O(f^2)$ relies on several new insights. 
We note that if the entire update sequence is known in advance 
and is stored in some data structure that allows for fast access, which is sometimes referred to as the ``(dynamic) offline setting'' (cf.\ \cite{CS18})---then a straightforward variant of our algorithm works  {\em deterministically} with $O(r^2)$ amortized update time. 
More specifically, whenever a matched edge is randomly sampled by our algorithm (which is always done uniformly) from a carefully chosen sample space of edges---in the offline setting, instead of randomly sampling the matched edge,
one can choose the matched edge {\em deterministically} to be the one that will be deleted last among all edges in the sample space. 
It is not difficult to verify that this simple tweak translates our probabilistic $O(f^2)$ amortized update time bound into a deterministic $O(f^2)$ time bound, while avoiding the entire   probabilistic analysis. Consequently, we believe it is instructive to defer the probabilistic ingredients of the analysis to the last sections (Sections \ref{app:r4} and \ref{app:r2}), and first present  the description of the algorithm and the non-probabilistic ingredients of its analysis.

%

\renewcommand{\HH}{\ensuremath{\mathcal{H}}}
\newcommand{\EE}{\ensuremath{\mathcal{E}}}
\newcommand{\VV}{\ensuremath{\mathcal{V}}}
\newcommand{\MM}{\ensuremath{\mathcal{M}}}
\renewcommand{\SS}{\ensuremath{\mathcal{S}}}

\newcommand{\MMs}{\ensuremath{\mathcal{M}^{\star}}}
\newcommand{\SSs}{\ensuremath{\mathcal{S}^{\star}}}

\section{Preliminaries}\label{sec:prelim}

\paragraph{Hypergraph Notation.} For a hypergraph $G = (V,E)$, $V$ denotes the set of vertices and $E$ denotes the set of hyperedges. We use $v \in e$ to mean that $v$ is one of the vertices incident on 
the hyperedge $e$.  
Rank $r$ of a hypergraph is the maximum number of vertices incident on each edge, i.e., $r := \max\set{ \card{e} \mid e \in E}$. 
 
 A {matching} $M$ in $G$ is a collection of vertex-disjoint hyperedges of $M$. A matching $M$ is called \emph{maximal} if no other edge of $G$ can be directly added to $M$ without violating its matching constraint. 
 A vertex cover $U$ in $G$ is a collection of vertices so that every hyperedge in $E$ has at least one endpoint in $U$. We have the following standard fact. 
 
\begin{fact}\label{fact:matching-vc}
	Let $G$ be any hypergraph with rank $r$. Suppose $M$ is a maximal matching in $G$ and $U$ denotes all vertices incident on $M$. Then $U$ is an $r$-approximate vertex cover of $G$. 
\end{fact}

\paragraph{Hypergraph Formulation of Set Cover.} Consider a family $\SS = \set{S_1,\ldots,S_m}$ of sets over a universe $[n]$. We can represent $\SS$ by a hypergraph $G$ on $m$ vertices corresponding to sets of $\SS$ and 
$n$ hyperedges corresponding to elements of $[n]$: Any element $i \in [n]$ is now a hyperedge between vertices corresponding to sets in $\SS$ that contain $i$. It is easy to see that there is a one-to-one correspondence between
set covers of $\SS$ and vertex covers of $G$ and that the rank $r$ of $G$ is the same as the maximum frequency parameter $f$ in the set cover instance. 

With this transformation, by~\Cref{fact:matching-vc}, obtaining an $f$-approximation to 
this instance of set cover reduces to obtaining a maximal matching of $G$. This is the direction we take in this paper for designing fully dynamic algorithms for set cover by designing a fully dynamic algorithm for hypergraph maximal matching. 
%

\newcommand{\staticalg}{\ensuremath{\textnormal{\texttt{static-sampling}}}\xspace}

\newcommand{\sample}{\textnormal{\ensuremath{\textsf{Sample}}}\xspace}

\newcommand{\OO}{\ensuremath{\mathcal{O}}}
\newcommand{\NN}{\ensuremath{\mathcal{N}}}
\renewcommand{\II}{\ensuremath{\mathcal{I}}}
\renewcommand{\AA}{\ensuremath{\mathcal{A}}}

\newcommand{\handlefree}{\ensuremath{\textnormal{\texttt{handle-free}}}\xspace}
\newcommand{\setowner}{\ensuremath{\textnormal{\texttt{set-owner}}}\xspace}
\newcommand{\setlevel}{\ensuremath{\textnormal{\texttt{set-level}}}\xspace}
\newcommand{\setlvl}{\ensuremath{\textnormal{\texttt{set-level}}}\xspace}
\newcommand{\detsettle}{\ensuremath{\textnormal{\texttt{deterministic-settle}}}\xspace}
\newcommand{\detr}{\ensuremath{\textnormal{\texttt{deterministic-settle}}}\xspace}
\newcommand{\randsettle}{\ensuremath{\textnormal{\texttt{random-settle}}}\xspace}
\newcommand{\rand}{\ensuremath{\textnormal{\texttt{random-settle}}}\xspace}

\newcommand{\ellold}{\ell^{\old}}
\newcommand{\ellnew}{\ell^{\new}}

\newcommand{\old}{\ensuremath{\textnormal{\textsf{old}}}}
\newcommand{\new}{\ensuremath{\textnormal{\textsf{new}}}}

\newcommand{\OOold}{\ensuremath{\mathcal{O}^{\old}}}
\newcommand{\OOnew}{\ensuremath{\mathcal{O}^{\new}}}
\newcommand{\IIold}{\ensuremath{\mathcal{I}}^{\old}}
\newcommand{\IInew}{\ensuremath{\mathcal{I}}^{\new}}

\newcommand{\tildeo}{\widetilde{o}}
\newcommand{\ellstar}{\ell^{*}}

\section{An $O(r^3)$-Update Time Algorithm}\label{sec:r4}

Throughout, we use $M$ to denote the maximal matching of the underlying hypergraph $G = (V,E)$ maintained by the algorithm. We use the following parameters in our algorithm: 
\begin{align}
	\alpha := (4 \cdot r), \qquad L := \lceil \log_{\alpha}\card{N} \rceil \label{eq:par},
\end{align}
where $N$ approximates the dynamic number of edges $|E|$ plus the fixed number of vertices $|V|$ from above, 
so that $\log_{\alpha} \card{N} = \Theta(\log_\alpha (|V| + |E|))$. 
Every $\Omega(N)$ steps we update the value of $N$, and as a result rebuild all the data structures; 
this adds a runtime of $O(|V| + |E|) = O(N)$ every $\Omega(N)$ update steps, hence a negligible overhead to the amortized cost of the algorithm. We may henceforth ignore this technical subtlety and treat $N$ as a fixed value in what follows.

\subsection{A Leveling Scheme and Hyperedge Ownerships}\label{sec:leveling-scheme}

We use a leveling scheme for the input hypergraph that partitions hyperedges and vertices. This is done by assigning a \emph{level} $\ell(e)$ to each hyperedge $e \in E$ and a \emph{level} $\ell(v)$ to each vertex $v \in V$. 
The leveling scheme 
satisfies the following invariant. 

\begin{invariant}\label{inv:level}
	Our leveling scheme satisfies the following properties: 
	\begin{enumerate}[label=(\roman*)]
		\item For any hyperedge $e \in E$, $0 \leq \ell(e) \leq L$ and for any vertex $v \in V$, $-1 \leq \ell(v) \leq L$; moreover, $\ell(v) = -1$ iff $v$ is \textbf{unmatched} by $M$. 
		\item For any \textbf{matched} hyperedge $e \in M$ and any incident vertex $v \in e$, $\ell(v) = \ell(e)$.
		\item For any \textbf{unmatched} hyperedge $e \notin M$, $\ell(e) = \max_{v \in e} \ell(v)$. 
	\end{enumerate}
\end{invariant}
Our leveling scheme needs only to specify $\ell(e)$ for each $e \in M$; the rest are fixed deterministically by~\Cref{inv:level}. Moreover, this invariant ensures that the matching $M$ obtained by the algorithm is maximal 
as all unmatched vertices are at level $-1$ while the level of any hyperedge is at least $0$ and at the same time equal to the maximum level of any of its incident vertices. 

Based on the leveling scheme, we assign each hyperedge $e$ to exactly one of its incident vertices $v \in e$ with $\ell(v) = \ell(e)$ to \emph{own} (the ties between multiple vertices at the same level are broken by the algorithm). 
We use $\OO(v)$ to denote the set of  edges owned by $v$. 

This definition, combined with \Cref{inv:level}, directly implies the following invariant  which we record  for the future references. 
\begin{invariant}\label{inv:record}
	${(i)}$ For any vertex $v \in V$, any owned hyperedge $e \in \OO(v)$, and any other incident vertex $u \in e$, $\ell(u) \leq \ell(v)$. ${(ii)}$ For any vertex $v \in V$, any incident hyperedge $e \in v$ has $\ell(e) \geq \ell(v)$. 
\end{invariant}

\subsection{Temporarily Deleted Hyperedges}\label{sec:deleted}

To obtain the desired update time of $O(r^3)$, we would need to allow some hyperedges of the hypergraph to be considered \emph{temporarily deleted} and no longer participate in any of other data structures; moreover, whenever we no longer consider them 
deleted, we simply treat them as a hyperedge insertion to our hypergraph and handle them similarly (which will take $O(r)$ time per each hyperedge exactly as in the previous algorithm). The role of these deleted hyperedges will become apparent 
only in the probabilistic analysis of~\Cref{app:r4}; however, it is the goal of the algorithm itself (rather than the analysis) to cope efficiently with the required deletions. We shall note that these deletions constitute one of several differences of our algorithm with that of~\cite{Sol16}. 

The following invariant allows us to maintain the maximality of our matching (and thus the correctness of our algorithm). 

\begin{invariant}\label{inv:deleted}	
	Any hyperedge that is \underline{temporarily deleted} is incident on some \textbf{matched} hyperedge. 
\end{invariant}

To maintain~\Cref{inv:deleted}, each matched hyperedge $e \in M$ is \emph{responsible} for a set of deleted edges denoted by $\DD(e)$ and stored in a linked-list data structure. 
This set will be \emph{finalized} at the time $e$ joins the matching $M$ and will remain unchanged throughout the algorithm until $e$ is removed from $M$; at that point, we simply bring back all hyperedges in $\DD(e)$
to the graph as new hyperedge insertions. 

We shall note that besides this data structure $\DD(e)$, these deleted hyperedges \emph{do not appear} in any other data structure of the algorithm and do not (necessarily) satisfy any of the
invariants in the algorithm -- they are simply treated as if they do not belong to the hypergraph. \Cref{inv:deleted} ensures that even though we are ignoring temporarily deleted hyperedges in the algorithm, the resulting maximal matching on the hypergraph of 
undeleted hyperedges is still a maximal matching for the entire hypergraph. As such, throughout the rest of the paper, with a slight abuse of notation, whenever we talk about hyperedges of $G$, we are referring to the hyperedges that are not temporarily deleted 
(unless explicitly stated otherwise).

\subsection{Data Structures}

We maintain the following information for each vertex $v \in V$ (again, to emphasize, we ignore the temporarily deleted hyperedges in all the following data structures): 
\begin{itemize}
	\item $\ell(v)$: the level of  $v$ in the leveling scheme; 
	\item $M(v)$: the hyperedge in $M$ incident on $v$ (if $v$ is unmatched $M(v) = \perp$);
	\item $\OO(v)$: the set of hyperedges $e$ owned by $v$ -- we define $o_v := \card{\OO(v)}$; 
	\item  $\NN(v)$: the set of hyperedges $e$ incident on $v$; 
	\item $\AA(v,\ell)$ for any integer $\ell \geq \ell(v)$: the set of hyperedges $e \in \NN(v)$ that are \emph{not} owned by $v$ and have level $\ell(e) = \ell$ -- we define $a_{v,\ell} := \card{\AA(v,\ell)}$.
\end{itemize}
We also maintain the following information for each hyperedge $e \in E$: 
\begin{itemize}
	\item $\ell(e)$:  the level of $e$ in the leveling scheme; 
	\item $O(e)$:  the single vertex $v \in e$ that owns $e$, i.e., $e \in \OO(v)$;
	\item $M(e)$: a Boolean variable to indicate whether or not $e$ is matched. 
\end{itemize}
We also maintain back-and-forth pointers 
between these different data structures that refer to the same hyperedge or vertex in a straightforward way.

In the following, we introduce the main procedures used for updating these data structures. In these procedures, whenever there is room for confusion, we use superscript $*^{\textsf{old}}$ to denote a parameter or data structure $*$ before the update and $*^{\textsf{new}}$ to denote the value of $*$ after the update. 
\begin{tbox}
\textbf{Procedure $\setowner(e,v)$.} Given a hyperedge $e $ and vertex $v \in e$ where $\ell(v) = \max_{u \in e} \ell(u)$, sets owner of $e$ to be $v$, i.e., $O(e) = v$.  
\end{tbox}
To implement $\setowner(e,v)$, we first set $O^{\new}(e) = v$ and $\ellnew(e) = \ell(v)$, add $e$ to $\OO(v)$, and remove $e$ from $\OO(O^{\old}(e))$. Then, if $\ellnew(e) = \ellold(e)$, there is nothing else to do. Otherwise, for any $w \neq v \in e$, 
we  remove $e$ from $\AA(w,\ellold(e))$ and instead insert $e$ in $\AA(w,\ellnew(e))$.

\begin{claim}\label{clm:set-owner}
	$\setowner(e,v)$ takes $O(r)$ time. 
\end{claim}
The proof of \Cref{clm:set-owner} is straightforward. 
\begin{tbox}
\textbf{Procedure $\setlevel(v,\ell)$.} Given a vertex $v \in V$ and integer $\ell \in \set{-1,0,\ldots,L}$, updates the level of $v$ to $\ell$, i.e., sets $\ell(v) = \ell$. 
\end{tbox}
The implementation of $\setlevel(v,\ell)$ is as follows. Firstly, we should determine the ownership of all hyperedges $e \in \OOold(v)$ that were previously owned by $v$. We will go over each hyperedge $e \in \OOold(v)$ one by one, find 
$u := \argmax_{w \in e} \ell(w)$ (where we use $\ellnew(v) = \ell$ for the computations here), and use the procedure $\setowner(e,u)$ to update the owner of $e$ to $u$.

If $\ellold(v) > \ell $, there is nothing else left to do as no new hyperedge needs to be owned by $v$ now that level of $v$ has decreased (\Cref{inv:record}).

If $\ellold(v) < \ell$, we should make $v$ the new owner of all hyperedges $e \in \NN(v)$ with level between $\ellold(v)$ to $\ell - 1$. This is done by traversing the lists maintained in 
$\AA(v,\ellold(v)), \cdots, \AA(v,\ell-1)$, and running $\setowner(e,v)$ for each edge $e$ in these lists. 

{\bf Remark.} It will be the responsibility of the procedure calling $\setlevel$ to make sure that the invariants (and particularly~\Cref{inv:level}) continue to hold.

\begin{claim}\label{clm:set-level}
	$\setlevel(v,\ell)$ takes $O\paren{r \cdot (o^{\old}_v + o^{\new}_v) + \ell+1}$ time. 
\end{claim}
\begin{proof}
	The algorithm iterates over all hyperedges in $\OOold(v)$ and use $\setowner$ that takes $O(r)$ time by \Cref{clm:set-owner} for each one. When $\ellold(v) > \ell$, this is all that is done by the algorithm and so the bound on runtime follows.
	In the other case, the algorithm also needs to iterate over the lists $\AA(v,\ellold(v)), \cdots, \AA(v,\ell-1)$ one by one which takes $O(\ell)$ time and run $\setowner$ for each of them that takes $O(r \cdot o^{\new}_v)$ time in total. 
	This gives the bound on runtime. 
\end{proof}

\subsection{The Update Algorithm}\label{sec:update}

There are multiple cases to handle by the update algorithm depending whether the updated hyperedge is an insertion or deletion, and whether or not it belongs to the maximal matching $M$. But the only interesting case is when a hyperedge in $M$
is deleted and that is where we start with. 

\subsubsection{Case 1: a hyperedge $e \in M$ is deleted} \label{sec:case-M-delete}

There are two things we should take care of upon deletion of a hyperedge $e$ from $M$; bringing back the temporarily deleted hyperedges in $\DD(e)$ to maintain~\Cref{inv:deleted}, 
and more importantly, updating the matching $M$ to ensure its maximality. All the interesting part of the algorithm happen in the second step, but before we get to that, let us quickly mention 
how the first part is done. 

\paragraph{Maintaining~\Cref{inv:deleted}.} Throughout the course of handling a single hyperedge update, we may need to delete multiple hyperedges $e_1,e_2,\ldots$ from $M$, starting from the originally 
deleted hyperedge by the adversary. 
Each of these hyperedge $e_i \in M$ that are now removed from $M$ is responsible for temporarily deleted hyperedges $\DD(e_i)$. 

We will maintain a queue of all hyperedges in $\DD(e_1),\DD(e_2),\ldots$ 
during the course of this update. At the \emph{very end} of the update, once all other changes are finalizes, we will insert the hyperedges in this queue one by one to the hypergraph as if they were inserted by the adversary 
(using the procedure of~\Cref{sec:case-insert}). This allows us to maintain~\Cref{inv:deleted}. 

We note that throughout the update we may also temporarily delete some new  hyperedges. We will show that in that case, all these hyperedges are also incident on some matched hyperedge which is responsible for them
and thus~\Cref{inv:deleted} holds for these hyperedges as well. 

\paragraph{Maintaining maximality of $M$.} Let us now begin the main part of the update algorithm. Suppose $e = (v_1,\ldots,v_r)$ is deleted from $M$. This makes the vertices $v_1,\ldots,v_r$ \emph{temporarily} free. We will handle each of these vertices using the procedure 
$\handlefree$ which is the key ingredient of the update algorithm (we will simply run $\handlefree(v_1),\ldots,\handlefree(v_r)$).

\begin{tbox}
\textbf{Procedure $\handlefree(v)$.} Handles a given vertex $v \in V$ which is \emph{unmatched} currently in the algorithm (its matched hyperedge may have been deleted via an update or by the algorithm in this time step, or $v$ may simply be
unmatched at this point of the algorithm)\footnote{To simplify the exposition, we may some time call $\handlefree(v)$ for a vertex
$v$ that is now matched again (while handling other vertices). In that case, this procedure simply aborts.}. 
\end{tbox}
The execution of Procedure $\handlefree(v)$ depends on the number of owned hyperedges of $v$, i.e., $o_v$ (both procedures used within this one are described below): $(i)$ if $o_v < \alpha^{\ell(v)+1}$, we will run $\detsettle(v)$; 
and otherwise $(ii)$ if $o_v \geq \alpha^{\ell(v)+1}$, we run $\randsettle(v)$ instead. We now describe each procedure.

\begin{tbox}
\textbf{Procedure $\detsettle(v)$.} Handles a given vertex $v \in V$ with $o_v < \alpha^{\ell(v)+1}$. 
\end{tbox}
In $\detsettle(v)$, we iterate over all hyperedges $e \in \OO(v)$ owned by $v$ and check whether all endpoints of $e$ are now unmatched; if so, we add $e$ to $M$, and run $\setlevel(v,0)$ and $\setowner(e,v)$. Moreover, for any 
vertex $u \neq v \in e$, we further run $\setlevel(u,0)$ so all vertices incident on $e$ are now at level $0$. If no such hyperedge is found,
we run $\setlevel(v,-1)$. 

\begin{claim}\label{clm:det-settle}
	$\detsettle(v)$ takes $O(r \cdot o^{\old}_v) = O(r \cdot \alpha^{\ellold(v)+1})$ time and maintains~\Cref{inv:level,inv:record} for vertex $v$ and hyperedges incident on $v$.
\end{claim}
\begin{proof}
	Checking if there is any hyperedge that can be added to $M$ takes $O(r \cdot o^{\old}_v)$ time. Moreover, running $\setlevel(v,0)$ or $\setlevel(v,-1)$ take $O(r \cdot o^{\old}_v)$ time by~\Cref{clm:set-level} as $o^{\new}_v \leq o^{\old}_v$ considering
	level of $v$ has not increased. 
	
	In case the algorithm finds a hyperedge $e$ to add to $M$, we run $\setlevel(u,0)$ for $u \neq v \in e$ as well which takes $O(1)$ time for each $u$ by~\Cref{clm:set-level} as $\ellold(u) = -1$ ($u$ was unmatched and by~\Cref{inv:level}) 
	and thus $o^{\old}_u = o^{\new}_u = 0$. As there are at most $r-1$ choices for $u$, this step takes $O(r) = O(r \cdot \alpha^{\ellold(v)+1})$ time (we are not being tight here). 
	We also need to run $\setowner(e,v)$ in this case which
	takes another $O(r)$ time by~\Cref{clm:set-owner}. Considering $\detsettle(v)$ is only called when $o^{\old}_v < \alpha^{\ellold(v)+1}$, we obtain the desired runtime bound. 
	
	As for maintaining~\Cref{inv:level}, consider the hyperedge $e$ chosen by the algorithm to be added to $M$. We set the level of all vertices incident on $e$ to be $0$ and making $v$ the owner of $e$, hence satisfying~\Cref{inv:level}.
	On the other hand, if we cannot find such a hyperedge $e$ incident on $v$, we know that all hyperedges incident on $v$ are already at level at least $0$ by~\Cref{inv:level} (even after removing $v$) and hence setting $\ellnew(v) = -1$ keeps~\Cref{inv:level}. 
	\Cref{inv:record} also holds by maintaining~\Cref{inv:level} and the fact that we choose correct owners in the algorithm. 
\end{proof}

Before defining $\randsettle$, we need the following definition. 
For any vertex $v \in V$ and integer $\ell > \ell(v)$, we define: 
\begin{itemize}
	\item $\tildeo_{v,\ell}$: the number of edges $v$ will own \emph{assuming} we increase its level from $\ell(v)$ to $\ell$. 
\end{itemize}
The following claim is a straightforward corollary of procedure $\setlevel$. 
\begin{claim}\label{clm:tildeo}
	For any $v \in V$ and $\ell > \ell(v)$, $\tildeo_{v,\ell} = \paren{\sum_{\ell'=\ell(v)}^{\ell-1} a_{v,\ell'}} + o_v$. In particular, $\tildeo_{v,\ell}$ can be obtained from $\tildeo_{v,\ell-1}$ in $O(1)$ time. 
\end{claim}
We can now describe $\randsettle$ (this is the main procedure of our update algorithm).
\begin{tbox}
\textbf{Procedure $\randsettle(v)$.} Handles a given vertex $v \in V$ with $o_v \geq \alpha^{\ell(v)+1}$. 
\end{tbox}

In $\randsettle(v)$, we first compute the level $\ellnew(v)$ as \emph{minimum} $\ell > \ell(v)$ with $\tildeo_{v,\ell} < \alpha^{\ell+1}$ and run $\setlevel(v,\ellnew(v))$. 
Then, we sample a hyperedge $e$ uniformly at random from $\OOnew(v)$. The next step of the algorithm now depends on this particular hyperedge $e$. 

\begin{itemize}
\item \textbf{Case (a): for all $u \in e$, $\tildeo_{u,\ellnew(v)} < \alpha^{\ellnew(v)+1}$.} In this case, we add the hyperedge $e$ to $M$ and run $\setlevel(u,\ellnew(v))$ for all $u \in e$ to maintain~\Cref{inv:level}. This potentially can make $M$ not a matching since some matched hyperedges $e_1,\ldots,e_k$ for $k < r$ incident on $e$ might be in $M$. We will thus delete these hyperedges from $M$ one by one and \emph{recursively} run the procedure of~\Cref{sec:case-M-delete} for each one, treating it \emph{as if} this hyperedge was 
deleted by the adversary (although we do \emph{not} remove the hyperedge from the hypergraph).

\item \textbf{Case (b): there exists  $u \in e$, s.t. $\tildeo_{u,\ellnew(v)} \geq \alpha^{\ellnew(v)+1}$.} We  run $\detsettle(v)$ to handle $v$ and ``switch'' the focus to $u$ instead. 
We then call $\setlevel(u,\ellnew(v))$. If $u$ is matched in $M$, say by hyperedge $e_u$, we  remove $e_u$ from $M$, and then \emph{recursively} run the 
procedure of~\Cref{sec:case-M-delete} for hyperedge $e_u$ -- we only note that when processing $e_u$, we start by running $\handlefree(u)$ first before all other vertices incident on $e_u$\footnote{This is only done for making the analysis more clear and is not 
needed.}. If $u$ is \emph{not} matched in $M$, we will simply run $\handlefree(u)$. 

\end{itemize}

\begin{claim}\label{clm:rand-settle}
	The first step of $\randsettle(v)$ before either case (changing level of $v$ and picking hyperedge $e$) takes $O(r \cdot \alpha^{\ellnew(v)+1})$ time. Additionally, case $(a)$ takes 
	$O(r^2 \cdot \alpha^{\ellnew(v)+1})$ time and case $(b)$ takes $O(r \cdot o^{\new}_u)$ time \underline{ignoring the recursive calls}. Finally, $\randsettle(v)$ maintains~\Cref{inv:level,inv:record}. 

\end{claim}
\begin{proof}
	Finding $\ellnew(v)$ takes $O(\ellnew(v))$ time by~\Cref{clm:tildeo} and $\setlevel(v,\ellnew(v))$ takes $O(r \cdot o^{\new}_v + \ellnew(v)) = O(r \cdot \alpha^{\ellnew(v)+1})$ 
	time by~\Cref{clm:set-level} as level of $v$ is increased (and so is number of its owned edges) and since $o^{\new}_v < \alpha^{\ellnew(v)+1}$ by design. This proves the first part. 
	
	The second part for case $(a)$ also follows because for any $u \in e$, $\setlevel(u,\ellnew(v))$ takes $O(r \cdot o^{\new}_u + \ellnew(u)) = O(r \cdot \tildeo_{u,\ellnew(v)} + \ellnew(v))$ by~\Cref{clm:set-level} 
	as level of $u$ is increased. This is $O(r \cdot \alpha^{\ellnew(v)+1})$ by the condition $\tildeo_{u,\ellnew(v)} < \alpha^{\ellnew(v)+1}$ in case (a). Multiplying this with at most $r$ vertices $u \in e$ gives 
	the desired bound. 

	The second part for case $(b)$ holds because $v$ has $o_v \leq \alpha^{\ellnew(v)+1}$ by the definition of $\ellnew(v)$ and thus $\detsettle(v)$ takes $O(r \cdot \alpha^{\ellnew(v)+1})$ time by~\Cref{clm:det-settle}. 
	Moreover, running $\setlevel(u,\ellnew(v))$ takes $O(r \cdot o^{\new}_u)$ by~\Cref{clm:set-level} which is at least $O(r \cdot \alpha^{\ellnew(v)+1})$ by the lower bound on $\tildeo_{u,\ellnew(v)}$ in this case. 
	
	Finally, \Cref{inv:level,inv:record} also hold as explained in the description of the procedure. 
\end{proof}

Before moving on, we record the following observation that plays a key rule in the recursive analysis of our algorithm. 

\begin{observation} \label{Ob:basic}
In $\randsettle(v)$: $(i)$ the hyperedge $e$ is sampled uniformly at random from at least $\alpha^{\ellnew(v)}$ edges. Moreover, $(ii)$ any hyperedge $e_1,\ldots,e_k$ or $e_u$ deleted from $M$ during this procedure is at level at most $\ellnew(v)-1$. 
\end{observation}

The first part of the observation is immediate by the definition of $\ellnew(v)$. For the second part, note that any deleted edge $e'$ in the process is incident on $e$ with $\ellold(e) = \ellold(v) < \ellnew(v)$. Let $u$ be any vertex incident on both $e'$ and $e$. 
Firstly, $\ell(u) \leq \ellold(v)$ as $e$ was owned by $v$ and not $u$, and secondly $\ell(e') = \ell(u)$ because $e'$ is a matched hyperedge (see~\Cref{inv:level}); as such $\ell(e') \leq \ell(e) < \ellnew(v)$ as well. 

\medskip

\paragraph{Temporarily deleting new hyperedges.} An astute reader may have noticed that we have not yet temporarily deleted any hyperedge in the update algorithm. The only place that this will be done in the algorithm is 
in case $(a)$ of $\randsettle$. 

In this case, when we decide to insert $e$ to $M$, we will temporarily delete all other hyperedges in $\OOnew(v)$ from which $e$ was sampled from, and make $e$ responsible for them by adding them to $\DD(e)$. 
The deletions of these hyperedges is done by the procedure of~\Cref{sec:case-noM-delete} (as they do not belong to $M$). As the cost of each such hyperedge deletion is $O(r)$ and their total number  
is $O(\alpha^{\ellnew(v)+1})$, this does not change the asymptotic bounds of~\Cref{clm:rand-settle}. Finally, since these edges are incident on $e$ which now belongs to $M$,~\Cref{inv:deleted} will be maintained by this process. 

The following observation is now immediate. 

\begin{observation}\label{obs:temp}
	Any hyperedge $e \in M$ is responsible for $O(\alpha^{\ell(e)+1})$ hyperedges in $\DD(e)$. 
\end{observation}

\subsubsection{Case 2: a hyperedge $e \notin M$ is deleted} \label{sec:case-noM-delete}

We only need to remove $e$ from the corresponding data structures which can be done in $O(r)$ time for the (at most) $r$ endpoints of $e$. 

\subsubsection{Case 3: a hyperedge $e$ is inserted} \label{sec:case-insert}

We need to find $v = \argmax_{u \in e} \ell(u)$ and run $\setowner(e,v)$ which takes $O(r)$ time. Moreover, if all vertices incident on $e$ are unmatched, we should additionally add $e$ to $M$ and run $\setlevel(u,0)$ for all $u \in e$ which takes $O(r)$ time per vertex by~\Cref{clm:set-level}. 

\medskip
\noindent
This concludes the description of our update algorithm.

%

\newcommand{\epoch}{\ensuremath{\textnormal{\textsf{epoch}}}\xspace}
\newcommand{\epochs}{\ensuremath{\textnormal{\textsf{epochs}}}\xspace}
\newcommand{\epochc}{\ensuremath{\epochs_{create}}\xspace}
\newcommand{\epocht}{\ensuremath{\epochs_{term}}\xspace}

\newcommand{\costc}{\ensuremath{C_{create}}\xspace}
\newcommand{\costt}{\ensuremath{C_{term}}\xspace}

\newcommand{\Nepoch}{\ensuremath{N}}

\renewcommand{\estar}{\ensuremath{e^{*}}}

\section{The Analysis of the $O(r^3)$-Update Time Algorithm} \label{analr3}

\subsection{Epochs}\label{sec:epochs}

A key definition in analysis is the notion of an \emph{epoch} borrowed from~\cite{BGS11,Sol16}. 

\begin{definition}[\textbf{Epoch}]
	For any time $t$ and any hyperedge $e$ in the matching $M$ at time $t$, \textbf{epoch} of $e$ and $t$, denoted by $\epoch(e,t)$, is the pair $(e,\set{t'})$ where $\set{t'}$ denotes the {maximal continuous time period} containing $t$ during
	which $e$ was {always present} in $M$ (not even deleted temporarily during one step and inserted back at the same time step). 
	
	\noindent
	We further define \underline{level} of $\epoch(e,t)$ as the level $\ell(e)$ of $e$ during the time period of the epoch\footnote{Note that our update algorithm never changes
	level of a hyperedge in $M$ without removing it from $M$ first and thus level of an epoch is well-defined.}. 
\end{definition}

The update algorithm in~\Cref{sec:update} takes $O(r)$ time deterministically for any update step that does \emph{not} change the matching $M$. However, an update at a time step $t$ that changes $M$ may force the algorithm 
a large computation time and hence we use amortization to charge the cost of processing such an update at time $t$ to the epochs that are created and removed at time $t$. In particular, for any time step $t$, define: 
\begin{itemize}
	\item $\epochc(t)$: the set of all $\epoch(e,t)$ that are \emph{created} at time $t$;  respectively, $\epocht(t)$ is defined  for epochs \emph{terminated} at time $t$;
	\item $\costc(\epoch(e,t))$: the computation cost at time $t$ for \emph{creation} of $\epoch(e,t)$;  respectively, $\costt(\epoch(e,t))$ is defined  for the cost of \emph{termination} of $\epoch(e,t)$. 
\end{itemize}
Then, the total cost of update at time $t$ is: 
\begin{align}
	C(t) = O(r) + \sum_{\substack{\epoch(e,t) \in \\ \epochc(t)}} \costc(\epoch(e,t)) + \sum_{\substack{\epoch(e,t) \in \\ \epocht(t)}} \costt(\epoch(e,t)). \label{eq:charge}
\end{align}

In the following two sections, we show how to charge the cost of each time step to the epochs created in that time step. The analysis in both these sections is deterministic. 
Finally, in~\Cref{sec:rand}, we use this charging scheme to bound the runtime of the algorithm using a probabilistic argument exploiting the random choices of the algorithm and obliviousness of the adversary. 

\subsection{Re-distributing Computation Costs to Epochs}\label{sec:det}

\begin{lemma}\label{lem:charge}
	The total computation cost $C(t)$ of an update at time $t$ in \Cref{eq:charge} can be charged to the epochs in the RHS so that any level-$\ell$ epoch is charged with $O(\alpha^{\ell+3})$ units of cost. 
\end{lemma}

We prove this lemma in a sequence of claims in this section. In the following, whenever we say ``some computation cost can be charged to $\costc(\epoch(e,t))$ or $\costt(\epoch(e,t))$'' we mean that the cost of computation is $O(\alpha^{\ell+3})$ and there 
is a level-$\ell$ epoch $\epoch(e,t)$ at time $t$ that is either created or terminated, respectively, and that this cost is charged to $\costc(\epoch(e,t))$ or $\costt(\epoch(e,t))$. We emphasize that each epoch during this process
is only charged a \emph{constant} number of times. 
This then immediately satisfies the bounds in~\Cref{lem:charge}. 

\begin{claim}\label{clm:charge-insert}
	Cost of a hyperedge-insertion update $e$ at time $t$ can be charged to $\costc(\epoch(e,t))$.
\end{claim}
\begin{proof}
	As in~\Cref{sec:case-insert}, this update takes $O(r)$ time and creates a level-$0$ $\epoch(e,t)$. 
\end{proof}

From now on, we consider the main case of a hyperedge-deletion update $e$ from $M$ at time $t$ using the procedure in~\Cref{sec:case-M-delete}. 
In this case, we remove $e$ from $M$ which results in $\epoch(e,t)$ to be terminated (and hence $\epoch(e,t) \in \epocht(t)$ in RHS of $C(t)$ in~\Cref{eq:charge}).
This step is then followed by running $\handlefree(v)$ for all $v \in e$. Recall that $\handlefree(v)$ can be handled either by $\detsettle(v)$ or $\randsettle(v)$ (depending on value of $o_v$). 

We first consider the easy case of $\detsettle(v)$.

\begin{claim}\label{clm:charge-delete-det}
	Let $e$ be a hyperedge in $M$ deleted at time $t$ (either by the adversary or the algorithm in a recursive call). Cost of $\handlefree(v)$ for all $v \in e$ that are handled by $\detsettle(v)$ can be charged to $\costt(\epoch(e,t))$.
\end{claim}
\begin{proof}
	By~\Cref{clm:det-settle}, for each  $v \in e$ that is handled by $\detsettle(v)$, the runtime of $\handlefree(v)$ is $O(r \cdot \alpha^{\ell(e)+1})$. As there are at most $r$ such vertices, 
	their total cost is $O(r^2 \cdot \alpha^{\ell(e)+1}) = O(\alpha^{\ell(e) + 3})$ by the choice of $\alpha$ in~\Cref{eq:par}. Hence the total cost of all these vertices that are charged to $\costt(\epoch(e,t))$ is 
	$O(\alpha^{\ell(e)+3})$. 
\end{proof}

Let us now switch to analyzing $\randsettle(v)$ for a fixed $v \in e$. The easier part here is when we can process $\randsettle(v)$ by case $(a)$. 

\begin{claim}\label{clm:charge-delete-case-a}
Cost of case $(a)$ of $\randsettle(v)$ at time $t$ \underline{ignoring the recursive calls} can be charged to $\costc(\epoch(e_v,t))$, where $e_v$ is the hyperedge inserted to $M$ in this case. 
\end{claim}
\begin{proof}
	In case $(a)$ of $\randsettle(v)$, we create a new edge $e_v$ in $M$ at level $\ell(e_v) = \ellnew(v)$. 
	By~\Cref{clm:rand-settle}, ignoring the recursive calls, this case takes $O(r^2 \cdot \alpha^{\ell(e_v)+1}) = O(\alpha^{\ell(e_v)+3})$ time. 
	We can thus charge the cost of $\randsettle(v)$ to $\costc(\epoch(e_v,t))$. 
\end{proof}

We now analyze case $(b)$ of $\randsettle(v)$. In this case, instead of handling $v$, we in fact recursively handle the vertex $u$ (with $\tildeo_{u,\ellnew(v)} \geq \alpha^{\ellnew(v)+1}$)  using $\handlefree(u)$ and only after that will come
back to take care of $v$ if needed. Moreover, since we first set level of $u$ to $\ellnew(v)$, we will have $o^{\new}_{u} \geq \alpha^{\ellnew(u)+1}$. This means $\handlefree(u)$ will be handled using $\randsettle(u)$ recursively. It is again possible
that $\randsettle(u)$ hits case $(b)$ and so on and so forth. However, note that when going from $v$ to $u$, we obtain that $\ellnew(v) < \ellnew(u)$ by~\Cref{Ob:basic} (here $\ellnew(u)$ refers to $\ell(u) = \ellnew(v)$ at the beginning of $\randsettle(u)$) and hence 
whenever case $(b)$ happens, the vertex we move on to will have a strictly larger level. As such, this chain of recursive calls to $\randsettle$ will terminate eventually in some case $(a)$ of $\randsettle$. 

We are going to denote the vertices in chain of  calls to $\randsettle$ on case $(b)$ for $v$ by $(v=)w_0,w_1,w_2,\ldots,w_k$ for some $k \leq L$ (number of levels); in particular, 
$w_1$ is the vertex $u$ in case $(b)$ of $\randsettle(v)$, $w_2$ is the vertex $u$ in case $(b)$ of $\randsettle(w_1)$, and so on. 

The computation cost of $\randsettle(v)$ in case $(b)$ then involves two parts: $(i)$ pre-processing before calling each $\randsettle(w_i)$ (which is bounded in~\Cref{clm:rand-settle}), $(ii)$ calling $\randsettle(w_i)$ itself, and $(iii)$ 
removing hyperedges $e_{w_i}$ from $M$ (for each one that exist) and recursively handling them using the procedure of~\Cref{sec:case-M-delete}. 

Among these costs, the cost of handling $e_{w_i}$ for $i \in [k]$ is handled separately by recursion (charged either to the termination of $\epoch(e_{w_i},t)$ or creation of new epochs). We thus need to handle the costs in parts $(i)$ and $(ii)$ in the following claim.

\begin{claim}\label{clm:charge-delete-case-b}
	Cost of case $(b)$ of $\randsettle(v)$ at time $t$ with chain of vertices $w_0,w_1,\ldots,w_k$ \underline{ignoring the recursive calls} (i.e., part $(iii)$ of costs above) can be charged to $\costc(\epoch(\estar,t))$, where $\estar$ is the hyperedge inserted to $M$
	in case $(a)$ of $\randsettle(w_k)$\footnote{Note that by definition, $\randsettle(w_k)$ finishes in case $(a)$.}.
\end{claim}
\begin{proof}
	By~\Cref{clm:rand-settle}, the sum of costs in part $(i)$ and $(ii)$ is:
	\begin{align*}
		\textnormal{total cost} &= O(r^2 \cdot \alpha^{\ellnew(w_k)+1}) + \sum_{i=1}^{k} O(r \cdot \alpha^{\ellnew(w_{i-1})+1}) + O(r \cdot o^{\new}_{w_i});
	\end{align*}
	In the above, $o_{w_i}$ and $\ell(w_i)$ for $i \in [k-1]$ refers to the size of ownership and level of $w_i$ \emph{before} running $\detsettle(w_i)$, while 
	$\ell(w_k)$ is the final level of $w_k$ (which is larger than the previous level of $w_k$). 
	
	Note that by the description of $\randsettle$, we will always have $o_{w_i} \leq \alpha^{\ell(w_i)+1}$ before calling the next $\randsettle(w_{i+1})$ 
	and so $O(r \cdot o_{w_i}) = O(r \cdot \alpha^{\ell(w_i)+1})$ (this can be a bit confusing as level of $w_i$ is changing \emph{twice} throughout this process; once in $\randsettle(w_{i-1})$ and then
	another time in $\randsettle(w_i)$; here $\ell(w_i)$ refers to the \emph{final} level of $w_i$). Another important note is that $\ell(w_i) \geq \ell(w_{i-1})+1$ again by the design of $\randsettle$. As such, 
	\begin{align*}
		\textnormal{total cost} &= O(r^2 \cdot \alpha^{\ellnew(w_k)+1}) + O(1) \cdot \paren{\sum_{i=1}^{k} r \cdot \alpha^{\ellnew(w_{i-1})+1} + r \cdot \alpha^{\ellnew(w_{i})+1}} \\
		&\leq O(r^2 \cdot \alpha^{\ellnew(w_k)+1}) + O(1) \cdot \paren{\sum_{i=1}^{k} 2 \cdot r \cdot \alpha^{\ellnew(w_{i})+1}} \\
		&\leq O(r^2 \cdot \alpha^{\ellnew(w_k)+1}) + O(1) \cdot 4 \cdot r \cdot \alpha^{\ellnew(w_{k})+1} \tag{as this is a geometric series}  \\
		&= O(\alpha^{\ellnew(w_k)+3}).
	\end{align*}
	As the final edge $\estar$ is created at level $\ell(w_k)$, we can charge the total cost to $\costc(\epoch(\estar,t))$, finalizing the proof. 
\end{proof}

Finally, we also have to handle the cost associated with maintaining~\Cref{inv:deleted}, namely, ``bringing back'' temporarily deleted hyperedges at the very end of the update step (the cost of temporarily deleting new hyperedges is accounted
for in~\Cref{clm:rand-settle} already). 

\begin{claim}\label{clm:temp-deleted}
	Cost of inserting back the temporarily deleted hyperedges in $\DD(e)$ for any hyperedge $e \in M$ can be charged to $\costt(\epoch(e,t))$. 
\end{claim}
\begin{proof}
	Any hyperedge $e \in M$ is responsible for $O(\alpha^{\ell(e)+1})$ hyperedges in $\DD(e)$ by~\Cref{obs:temp}. As bringing back these hyperedges requires $O(r \cdot \alpha^{\ell(e)+1})$ time in total, 
	this charge can be charged to the termination of $\epoch(e,t)$. 
\end{proof}

We thus showed that any cost in $C(t)$ can be charged by a factor of $O(\alpha^{\ell+3})$ to some level-$\ell$ epoch in RHS of~\Cref{eq:charge} without charging any epoch more than a constant number of times, 
thus finalizing the proof of~\Cref{lem:charge}. 

\subsection{Natural and Induced Epochs} 

\Cref{lem:charge} allows us to charge the computation costs in each time step of the algorithm (in~\Cref{eq:charge}) to the epochs created and terminated in that time step. However, 
for the rest of our analysis to work, we need to consider a ``special'' type of epochs and charge all costs only to those epoch. This section is dedicated to this purpose. 

To continue, we are going to make the following assumption for the rest of the analysis. 

\begin{assumption}\label{assumption1}
	The set of updates ends $G$ in an \emph{empty} hypergraph. In particular, every $\epoch$ will be terminated at some point in the algorithm. 
\end{assumption}
\Cref{assumption1} is without loss of generality as we can artificially append 
a set of hyperedge-deletions at the end of the original adversarial update sequence, so as to finish the sequence with an empty hypergraph. This can only increase the total time of the algorithm
while increasing the number of updates by a factor of at most two, which does not affect the amortized update time asymptotically. 

In~\Cref{lem:charge}, we charge at most $O(\alpha^{\ell+3})$ to the creation and/or termination of a level-$\ell$ epoch. As under~\Cref{assumption1} every epoch created will also be terminated, 
we can re-distribute the charges to creation of each level-$\ell$ epoch to the termination of the same epoch. As such, we now have an equivalent charging scheme in which \emph{only termination} of each level-$\ell$ epoch is charged 
with $O(\alpha^{\ell+3})$ computation time (and not the creations).

Recall that an epoch is terminated when the corresponding hyperedge $e$ is deleted from the maximal matching $M$. There are two types of hyperedge deletions from $M$ in the algorithm: the ones that are a result of the adversary deleting a hyperedge from the graph, and the ones that are the result of the update algorithm to remove a matched hyperedge in favor of another (so the original hyperedge is still part of the hypergraph).  Based on this, we differentiate between epochs as follows: 
\begin{itemize}
	\item \textbf{Natural epoch:} We say an $\epoch(e,*)$ is \emph{natural} if it {ends} with the adversary deleting the hyperedge $e$ from $G$. 
	\item \textbf{Induced epoch:} We say $\epoch(e,*)$ is \emph{induced} if it ends with the update algorithm removing the hyperedge $e$ from $M$ (so $e$ is still part of $G$). 
\end{itemize}
We are now going to re-distribute the charge assigned to all induced epochs between the natural epochs. This then will allow us to focus solely on natural epochs. 

\begin{lemma}\label{lem:induced-natural}
	The total cost charged to all induced epochs can be charged to natural epochs so that any level-$\ell$ natural epoch is charged with the costs of at most $r-1$ induced epochs at \emph{lower levels}.
\end{lemma}
\begin{proof}
	We will go over epochs from time step $1$ till the end and in each time step and either charge the cost of an induced epoch terminating at this time step to some higher-level natural epoch terminating at this time, or 
	to another higher-level epoch (natural or induced) that is created at this time step -- these latter charges are then re-distributed once these epochs are terminated (in case they are induced ones). 
	
	Consider any time step $t$ in the algorithm that involves deletion of a hyperedge from $M$. The induced epochs at this time step (if any) are the results of running $\randsettle$. In particular, in case $(a)$ of 
	$\randsettle(v)$, we remove hyperedges $e_1,\ldots,e_k$ for $k \le r-1$ from $M$ that are incident on the hyperedge $e_v$ chosen for $v$; and in case $(b)$, we remove the hyperedge $e_u$ from $M$ which is the matching 
	hyperedge of the special vertex $u$ in this case plus the hyperedges removed from $M$ during the recursive call to $\randsettle(u)$. 
	
	Recall that in~\Cref{clm:charge-delete-case-a,clm:charge-delete-case-b}, the cost associated with $\randsettle(u)$, ignoring the recursive calls, is charged to epochs \emph{created} in this time step. 
	Hence, in the current argument, 
	we do not need to worry about those. The only costs remained to consider are thus the ones in recursive calls that are charged to termination of epochs (corresponding to edges) $e_1,\ldots,e_k$ via~\Cref{clm:charge-delete-det}; these charges are exactly 
	the ones for $\detsettle(w)$ for $w \in e_1,\ldots,e_k$.
	Consider any $e_i \in \{e_1,\ldots,e_k\}$. We have $\ell(e_i) \leq \ell(v)-1$ by~\Cref{Ob:basic}. 
		We charge the 
		termination of each of these $k \le r-1$ epochs to whichever epoch creation (at level $\ell(v)$) the cost of $\randsettle(v)$ itself
		(in~\Cref{clm:charge-delete-case-a})
		is charged to, finalizing the proof. 
\end{proof}

The next step is to move from natural and induced epochs to the corresponding notions in \emph{levels}: We say that a level $\ell$ is an \textbf{induced level} (respectively, \textbf{natural level})
if the number of induced level-$\ell$ epochs  is greater than (resp., at most) the number of natural level-$\ell$ epochs.

Next, we will charge the computation costs incurred by any induced level to the computation costs at higher levels, so that in the end, the entire cost of the algorithm will be charged to natural levels.
Specifically, in any induced level $\ell$, we define a one-to-one mapping from the natural to the induced epochs.
For each induced epoch, at most one natural epoch (at the same level) is mapped to it; any natural epoch that is mapped to an induced epoch is called \emph{semi-natural}.
For any induced level $\ell$, all natural $\ell$-level epochs are semi-natural by definition.  
For any natural level, all natural epochs terminated at that level remain as before; these epochs are called \emph{fully-natural}.

By~\Cref{lem:induced-natural}, for any  epoch, at most $r-1$ induced epochs at lower levels are charged to it. 
We  define the \textbf{recursive cost} of an epoch as the sum of its actual cost and the recursive costs of the at most $r-1$ induced epochs charged to it as well as the (at most) $r-1$ semi-natural epochs mapped to them; the recursive cost of a level-0 epoch is defined as its actual cost.

Denote the highest recursive cost of any level-$\ell$ epoch by $\hat C_\ell$, for any level $\ell \ge 0$;  $\hat C_\ell$ is monotone non-decreasing with $\ell$.
By~\Cref{lem:charge}, the actual cost of a level-$\ell$ epoch is $O(\alpha^{\ell+3})$, hence
we get the recurrence:  
\[
\hat C_{\ell} \le   O(\alpha^{\ell+3}) + \sum_{i=1}^{k} 2 \cdot \hat  C_{\ell_i},
\]
where $\ell_i < \ell$ for $i \in [k]$ and $k \leq r-1$. By the monotinicty of $\hat C_\ell$, this results in: 
\[
	\hat C_{\ell} \le 2\,(r-1) \cdot \hat C_{\ell-1} + O(\alpha^{\ell+3}). 
\]

The base condition of this recurrence is $\hat C_{0} =  O(\alpha^3)$, hence the recurrence resolves to $\hat C_{\ell} = O(\alpha^{\ell+3})$ 
for $\alpha \geq 3 \cdot (r-1)$. We thus conclude the following lemma.

\begin{lemma} \label{important}
For any $\ell \ge 0$, the recursive cost of any level-$\ell$ epoch is bounded by $O(\alpha^{\ell+3})$.
\end{lemma}

The sum of recursive costs over all fully-natural epochs is equal to the sum of actual costs over all  epochs (fully-natural, semi-natural and induced)
that have been \emph{terminated} throughout the update sequence,
which, by Assumption~\ref{assumption1}, bounds the total runtime of the algorithm.

\subsection{Runtime Analysis}\label{sec:rand}

The final step of the analysis is to use the randomization in the algorithm (and obliviousness of the adversary) to prove 
a probabilistic upper bound on the amortized cost of the algorithm.

Let $Y$  be the r.v.\ for the sum of recursive costs over all fully-natural epochs terminated during the entire sequence. As argued in~\Cref{important},
$Y$will be equal to the total runtime of the algorithm.  
\begin{lemma}  \label{expect}
Let $t$ denote the number of updates in the algorithm (under~\Cref{assumption1}). With high probability, $Y = O(t \cdot \alpha^3 + N \log N \cdot \alpha^3)$. 
\end{lemma}  

Lemma \ref{expect} should be compared to Lemma 3.5 of \cite{Sol16}, which states that
$Y = O(t + n \log n)$ w.h.p. in an ordinary graph. The proof of Lemma 3.5 of \cite{Sol16} (given in the full version of the paper~\cite[Lemma 4.6 and Appendix B]{Solom16}) carries over rather smoothly to our case but with a few subtle tweaks. 
We provide the details in \Cref{app:r4}.

The high probability proof easily extends (as in \cite{Sol16}) to obtain $\Expect(Y) = O(t \cdot \alpha^3 + N \log N \cdot \alpha^3)$ considering that the runtime of the algorithm can be bounded by some $\poly(N)$ in the worst case and deterministically.
Shaving the $O(N \log N \cdot \alpha^3)$ term from the high probability bound, and consequently from the expected bound, requires more work, but this too can be done as in \cite{Solom16}.

This allows us to conclude the following theorem. 

\begin{theorem}
For any integer $r > 1$, starting from an empty rank-$r$ hypergraph on a fixed set of vertices, a maximal matching can be maintained 
over any sequence of $t$ hyperedge insertions and deletions
in $O(t \cdot r^3)$ time in expectation and with high probability. 
\end{theorem}

%

\newcommand{\WW}{\ensuremath{\mathcal{W}}}
\newcommand{\PP}{\ensuremath{\mathcal{P}}}

\newcommand{\CC}{\ensuremath{\mathcal{C}}}

\newcommand{\testinsert}{\ensuremath{\textnormal{\textsf{test-and-insert}}\xspace}}

\newcommand{\case}[1]{\ensuremath{\textnormal{\textbf{Case #1}}}\xspace}

\section{An $O(r^2)$-Update Time Algorithm}\label{sec:r2}

We now give our improved and main algorithm which heavily builds on and extend our previous $O(r^3)$ update time algorithm. 
Throughout this section, we use the same notation and parameters as before unless explicitly stated  otherwise. 
Moreover, we will also use the same  leveling scheme and its corresponding invariants from \Cref{sec:leveling-scheme}. The same also holds 
for the definition of temporarily deleted hyperedges (albeit they will be handled rather differently in this algorithm). 

\subsection{Data Structures}

We maintain the same information for each vertex $v \in V$ and each hyperedge $e \in E$ as before. 
Additionally, we maintain the following data structures for each vertex $v \in V$:  

\begin{itemize}
	\item $\SS(v)$: For any owned hyperedge $e \in \OO(v)$, we say $e$ is a \emph{strong} hyperedge for $v$ iff there exists $u \neq v \in e$ such that $\ell(u) = \ell(v) (=\ell(e))$; we store 
	all strong hyperedges of $v$ in a set $\SS(v)$ and define $s_v := \card{\SS(v)}$. 
	\item $\WW(v)$: Any hyperedge $e \in \OO(v)$ that is not a strong hyperedge for $v$ is called a \emph{weak} hyperedge and is stored in a set $\WW(v) := \OO(v) \setminus \SS(v)$ -- 
	we further define $w_v := \card{\WW(v)}$. 
\end{itemize}

\noindent
The following data structure will also be kept for each hyperedge $e \in E$: 

\begin{itemize}
	\item $\PP(e)$: List of all vertices $v \in e$ such that $\ell(v) = \ell(e)$ referred to as \emph{potential owners} of $e$; recall that $O(e)$ is the vertex owning edge $e$, and notice that $O(e) \in \PP(e)$ but $\PP(e)$ may contain other vertices as well (all of which are endpoints of $e$).
\end{itemize}

The main procedures for maintaining the validity of our data structures were $\setowner(e,v)$ and $\setlevel(v,\ell)$. We recall the description of these data structures in the following and mention
their new implementation. 

\begin{tbox}
\textbf{Procedure $\setowner(e,v)$.} Given a hyperedge $e $ and vertex $v \in e$ where $\ell(v) = \max_{u \in e} \ell(u)$, sets owner of $e$ to be $v$, i.e., $\OO(e) = v$.  
\end{tbox}

To implement $\setowner(e,v)$, we first set $O^{\new}(e) = v$ and $\ellnew(e) = \ell(v)$, add $e$ to $\OO(v)$, and remove $e$ from $\OO(O^{\old}(e))$. 
Moreover, if $\ellnew(e) \neq \ellold(e)$, then we go over all $u \in e$ and we remove $e$ from $\AA(u,\ellold(e))$ and instead insert $e$ in $\AA(u,\ellnew(e))$. This part is exactly as before. 

We now explain the new changes. We check if $v$ is the only potential owner of $e$ or not by looking at the level of any vertex $u \neq v \in \PP(e)$: $(i)$ if $\ell(u) = \ellnew(v)$, then 
all other vertices in $\PP(e)$ remain potential owners of $v$ and we only need to add $e$ to $\SS(v)$ as it is a strong hyperedge for $v$. Otherwise, $(ii)$ if $\ell(u) < \ellnew(v)$, then
no vertex incident on $e$ is at the same level at $v$ and so we empty $\PP(e)$ and insert $e$ into $\WW(v)$ as a weak hyperedge for $v$.

\begin{claim}\label{clm:set-owner-r2}
	$\setowner(e,v)$ takes $O(1)$ time if $\ellold(e) = \ell(v)$ and $O(r)$ time otherwise. 
\end{claim}
The proof of \Cref{clm:set-owner-r2} is straightforward and is thus omitted. Note that the main difference between this claim and~\Cref{clm:set-owner} is that now we are bounding the runtime of $\setowner(e,v)$ when level of $v$ is the same
as the current level of $e$ by just $O(1)$ time as opposed to $O(r)$. This will be crucial for the improvement in our algorithm. 

\begin{tbox}
\textbf{Procedure $\setlevel(v,\ell)$.} Given a vertex $v \in V$ and integer $\ell \in \set{-1,0,\ldots,L}$, updates the level of $v$ to $\ell$, i.e., sets $\ell(v) = \ell$. 
\end{tbox}
The new implementation of $\setlevel(v,\ell)$ is as follows. 

If $\ellold(v) > \ell $, we first go over all strong hyperedges $e \in \SS^\old(v)$ and switch their owners to some arbitrary vertex $u \neq v \in \PP(e)$ by running $\setowner(e,u)$. We then 
go over each weak hyperedge $e \in \WW^\old(v)$ one by one and compute $u := \argmax_{z \in e} \ell(z)$ (where we use $\ellnew(v) = \ell$ for the computations here) and use the procedure $\setowner(e,u)$ to update the owner of $e$ to $u$. 
Note that no new hyperedge needs to be owned by $v$ in this case as the level of $v$ has decreased. 

If $\ellold(v) < \ell$, all hyperedges in $\OOold(v)$ remain in the ownership of $v$. However, we need to set their level to $\ellnew(v) = \ell$; this is done by iterating over all $e \in \OOold(v)$ and running $\setowner(e,v)$ again which fixes their levels. 
At this point, we should also make $v$ the new owner of all hyperedges $e \in \NN(v)$ with level between $\ellold(v)$ to $\ell - 1$. This is done by traversing the lists maintained in 
$\AA(v,\ellold(v)), \cdots, \AA(v,\ell-1)$, and running $\setowner(e,v)$ for each hyperedge $e$ in these lists. 

The following claim is a more nuanced variant of~\Cref{clm:set-level} tailored to the new algorithm. 

\begin{claim}\label{clm:set-level-r2}
	$\setlevel(v,\ell)$ takes:  
	\begin{enumerate}[label=$(\roman*)$]
		\item $O(s^\old_v + r^\old_v \cdot r)$ time when $\ellold(v) > \ell$;
		\item $O(o^{\new}_v \cdot r + \ell)$ time when $\ellold(v) < \ell$.
	\end{enumerate}
\end{claim}
\begin{proof}
	The first part holds by~\Cref{clm:set-owner-r2} as for the strong hyperedges $e \in \SS^{\old}(v)$, the level of $e$ does not change in this process and thus it only takes $O(1)$ time to process them as opposed to $O(r)$. 
	
	The second part holds because $o^{\new}_v \geq o^{\old}_v$ and we spend $O(r)$ time per each hyperedge in $\OO^{\new}(v)$ by~\Cref{clm:set-owner-r2}; the extra $O(\ell)$ term now accounts for the case
	when some of the lists in $\AA(v,\ellold(v)), \cdots, \AA(v,\ell-1)$ are empty and thus we simply iterate over their indices. 
\end{proof}

\subsection{The Update Algorithm}\label{sec:update-r2}

We now present the update algorithm. As before, the main case of the algorithm is when a hyperedge $e$ is deleted from the maximal matching. In fact, the rest is 
exactly as before and are thus not repeated. Moreover, exactly as in the $O(r^3)$-update time algorithm, we handle the temporarily deleted hyperedges 
by simply storing them explicitly throughout the update and only inserting them back at the very end of this update phase. We will discuss temporarily deleting new hyperedges in the algorithm 
in each process that performs them.

We are now ready to present the update algorithm for the case when $e \in M$ is deleted. 
The single most important change from our previous algorithm is the ``hyperedge-centric'' view of the new algorithm versus ``vertex-centric'' in the previous case. In particular, we now handle each hyperedge as a single 
unit as opposed to handling each of its vertices separately (which was done through the $\handlefree$ procedure). For this purpose, we introduce the following new procedure.

\begin{tbox}
\textbf{Procedure $\handleedge(e)$.} Handles a given hyperedge $e$ which got deleted from the hypermatching $M$ (it might be that $e$ is deleted entirely  from the hypergraph or just from the matching; the treatment is the same in both cases). 
\end{tbox}

Before getting to describe $\handleedge(e)$, we need some definitions and auxiliary procedures that are described below.

We use a similar procedure as in $\detsettle$ for the previous algorithm in this case also. We shall however note that the criteria for calling this procedure in $\handleedge(e)$ is different than before and in particular is \emph{not} a function of $v$ alone. 

\begin{tbox} 
\textbf{Procedure $\detsettle(v)$.} Handles a vertex $v \in e$ as part of $\handleedge(e)$. 
\end{tbox}
In $\detsettle(v)$, we iterate over all hyperedges $e' \in \OO(v)$ owned by $v$ and check whether all endpoints of $e$ are now unmatched; if so, we add $e$ to $M$, and run $\setlevel(v,0)$ and $\setowner(e,v)$. Moreover, for any 
vertex $u \neq v \in e$, we further run $\setlevel(u,0)$ so all vertices incident on $e$ are now at level $0$. If no such hyperedge is found,
we run $\setlevel(v,-1)$. 

The following~\Cref{clm:det-settle-r2} is an analogue of~\Cref{clm:det-settle} (with the difference that the runtime now is  a factor $O(r)$ smaller in some cases). 

\begin{claim}\label{clm:det-settle-r2}
	$\detsettle(v)$ takes $O(s_v + r \cdot w_v)$ time and maintains~\Cref{inv:level,inv:record} for vertex $v$ and hyperedges incident on $v$ other than $e$.
\end{claim}
\begin{proof}
	The cost of going over hyperedges of $e \in \OO(v)$ can be partitioned into two parts: for edges in $\SS(v)$, we do not need to explicitly check whether all endpoints of $e$ are unmatched, 
	since being in $\SS(v)$ implies that they have a potential owner and thus this potential owner is matched; as such, we can ignore these edges at this point. For the remaining edges, the runtime is $O(r \cdot w_v)$, so
	the total runtime is $O(s_v + r \cdot w_v)$. Moreover, since $\ell(v)$ is decreasing in this case, by~\Cref{clm:set-level-r2}, the runtime of $\setlevel(v,0)$ or $\setlevel(v,-1)$ is also $O(s_v + r \cdot w_v)$.

	The second part about maintenance of~\Cref{inv:level,inv:record} is exactly as in~\Cref{clm:det-settle}. 
\end{proof}

The following definition is a key part of $\handleedge$ which tells us the criteria for different ways of handling $e$. 

\begin{definition}[\textbf{Tentative cost}]\label{def:tc}
	For a hyperedge  $e$, we define the \emph{tentative cost} of $e$ as: 
	\[
	\Gamma(e) := \max\set{\sum_{v \in e} s_v, r \cdot \sum_{v \in e} w_v}.
	\]
	We call $e$ \emph{strong-costly} if $\Gamma(e)$ is equal to the first term of the max and otherwise call it \emph{weak-costly}. 
\end{definition}

We will use a simple procedure that determines whether a hyperedge can be inserted to the matching based on our tentative cost criteria or not. 

\begin{tbox} 
\textbf{Procedure $\testinsert(e)$.} Insert a hyperedge $e$ to $M$ or decides it is too costly to do so. 
\end{tbox}
\noindent
The procedure consists of the following two cases:  
\begin{itemize}
\item \textbf{Case (a): when $\Gamma(e) < 1000 \cdot \alpha^{\ell(e)+2}$.} In this case, we will insert $e$ in $M$ and remove any hyperedge 
$e'$ incident on $e$ that is part of $M$. This is done by deleting them one by one from $M$ and  \emph{recursively} running $\handleedge$ for each one. 

\item \textbf{Case (b): when $\Gamma(e) \geq 1000 \cdot \alpha^{\ell(e)+2}$.} In this case, we do not insert $e$ in $M$ and instead run $\handleedge(e)$. 

\end{itemize}

Finally, we also use a procedure similar-in-spirit to $\randsettle$ of the previous algorithm. The difference here however is that we sample a hyperedge to join the matching from a ``large'' set of hyperedges 
that do not necessarily belong to the ownership of a single vertex. This procedure and its (probabilistic) analysis is the single most important difference of our new algorithm with our own $O(r^3)$-time algorithm as well as prior
work in~\cite{BGS11,Sol16} for ordinary graphs.

\begin{tbox} 
\textbf{Procedure $\addrandom(S)$.} Inserts a random hyperedge to $M$ from a  set $S$ of hyperedges.
\end{tbox}
The implementation of $\addrandom(S)$ is as follows. Let $\ellstar := \argmax_{\ell} \card{S} \geq \alpha^{\ell}$. We sample 
a hyperedge $e \in S$ uniformly at random, set $\ell(e) = \ellstar$ and $\setlevel(v,\ellnew(e))$ for all $v \in e$ and finally run $\testinsert(e)$. 

\paragraph{Temporarily deleting new hyperedges.} In the new algorithm, this is the place where we are going to temporarily delete some new hyperedges. 
Basically, we will add any hyperedge $e' \in S$ which is incident on $e$ to $\DD(e)$ for $e$ to be responsible for, and temporarily delete them from the hypergraph. We shall emphasize a key subtlety here: while in our previous algorithm, 
the set of edges that we were sampling a hyperedge $e$ from were all incident on $e$ (in $\randsettle$), this is \emph{no longer} the case for the new algorithm. Consequently, we are \emph{not} allowed to temporarily delete \emph{all} of $S$
from the graph (as it will violate~\Cref{inv:deleted}) but only the ones that are incident on $e$.

\begin{claim}\label{clm:addrandom-r2}
	Suppose $\ellnew(e) \geq \ellold(e)$; then $\addrandom(S)$ takes $O(\Gamma^{\new}(e) + \alpha^{\ellnew(e)+2})$ time. Moreover, it maintains \Cref{inv:level,inv:record}.  
\end{claim}
\begin{proof}
	Sampling $e$  takes $O(\card{S}) = O(\alpha^{\ellstar+1}) = O(\alpha^{\ellnew(e)+1})$ time. 
	Running $\setlevel(v,\ellnew(e))$ for all $v \in e$ takes $O(r \cdot o^{\new}_v + \ellnew(e))$ time by~\Cref{clm:set-level-r2} as level of $v$ is increased during this process since $\ell(v) \leq \ellold(e)$ (by~\Cref{inv:level}) and $\ellold(e) \leq \ellnew(e)$ (by the 
	assumption). As such, the total cost of this step is asymptotically: 
	\[
		\sum_{v \in e} (r \cdot o^{\new}_v + \ellnew(e)) = (\sum_{v \in e} r \cdot w^{\new}_v) + r \cdot \ellnew(e) \leq \Gamma^{\new}(v) + \alpha^{\ellnew(e)+2},
	\]
	as any vertex $v$ that increases its level will have $s_v = 0$ since it does not own any hyperedge at its own level (by definition of $\setlevel$). 
	
	Finally, temporarily deleting remaining edges from $S$ takes $O(r \cdot \card{S}) = O(\alpha^{\ellnew(e)+2})$ time, which finalizes the total time bound. 
	
	\Cref{inv:level,inv:record} also hold as explained throughout the description of the procedure. 
\end{proof}

\paragraph{Implementation of $\handleedge$.} We are now ready to describe the implementation of $\handleedge(e)$. The procedure involves multiple cases which we go over one by one. We note that even for a single hyperedge $e$, 
after processing $e$ in one case, the algorithm may switch to another case; however, as will be evident, this can only happen from cases with larger numbers to cases with smaller ones and hence there is no loop in this part of the algorithm. 

\subsubsection{Case 1: when $\Gamma(e) < 1000 \cdot \alpha^{\ell(e)+2}$.} \label{sec:small-r2}

This is an easy case. We will simply run $\detsettle(v)$ for all $v \in e$ which will ensure that all invariants are satisfied. 
The following claim bounds the runtime of the procedure. 

\begin{claim}\label{clm:case1-r2}
	The runtime of $\handleedge(e)$ in \case{1} is $O(\Gamma(e)) = O(\alpha^{\ell(e)+2})$. 
\end{claim}
\begin{proof}
	By~\Cref{clm:det-settle-r2}, we have that the total runtime is:
	\begin{align*}
		O(1) \cdot \sum_{v \in e} (s_v + r \cdot w_v) = O(1) \cdot \max \set{\sum_{v \in e} s_v, r \cdot \sum_{v \in e} w_v} = O(\Gamma(e)). 
	\end{align*}
	The bound on  $\Gamma(e)$ in this case finalizes the proof. 
\end{proof}

\subsubsection{Case 2: when $\Gamma(e) \geq 1000 \cdot \alpha^{\ell(e)+2}$ and $e$ is strong-costly}\label{sec:strong-costly-r2}

This case is also relatively easy and is similar to the previous algorithm. 

Let $v := \argmax_{u \in e} s_u$. Since $e$ is strong-costly and by the lower bound on $\Gamma(e)$, we have that 
\[
s_v \geq 1000 \cdot \alpha^{\ell(e)+1}.
\]

We first go over all vertices $u \neq v \in e$ and run $\detsettle(u)$ to handle these vertices. So it only remains to handle $v$ which is done as follows. 

Recall the definition of $\tildeo_{v,\ell}$ as the number of edges $v$ \emph{would own assuming that}  we increase its level to $\ell > \ell(v)$. We first find $\ell := \argmax_{\ell} \set{\tildeo_{v,\ell} > \alpha^{\ell}}$. 
We then run $\setlevel(v,\ell)$ followed by $\addrandom(\OOnew(v))$ which inserts the hyperedge $e_v$ to $M$ -- this hyperedge will be kept in $M$ according to the criteria of $\testinsert$ or will be deleted
subsequently and handled via $\handleedge$ recursively.

This finalizes the description of the update algorithm in this case.

\begin{claim}\label{clm:case2-r2}
	$\handleedge(e)$ in \case{2} takes $O(\Gamma^{\new}(e_v) + \alpha^{\ellnew(e_v)+2})$ time \underline{ignoring the recursive calls}.  
\end{claim}
\begin{proof}
	Running $\detsettle(u)$ for all $u \neq v \in e$ takes $O(\Gamma(e))$ time by the same analysis as in~\Cref{clm:case1-r2}. At the same time, since $e$ is strong-costly, 
	\begin{align*}
		\Gamma(e) = \sum_{u \in e} s_u \leq r \cdot s_v \leq r \cdot \tildeo_{v,\ellnew(v)} \leq \alpha^{\ellnew(v)+2} = \alpha^{\ellnew(e_v)+2},
	\end{align*}
	and so the runtime of this part is within the budget. Also since $\ellnew(e_v) = \ellnew(v) > \ellold(e)$,~\Cref{clm:addrandom-r2} implies that runtime of $\addrandom$ is $O(\Gamma^{\new}(e_v)+\alpha^{\ellnew(e_v)+2})$. 
	The remainder of the algorithm in this case is only recursive calls that is ignored in this claim.  

\end{proof}

We record the following observation for this case. 

\begin{observation}\label{obs:case2-r2}
	$(i)$ For the vertex $v$, $\ellnew(v) > \ellold(v)$. $(ii)$ For the hyperedge $e_v$ added to $M$, $e_v$ is sampled from a set of at least $1000 \cdot \alpha^{\ellnew(e)}$ hyperedges that are all incident on $e$. $(iii)$ For any hyperedge $e'$ deleted from $M$ 
	in case (a) of $\testinsert(e_v)$ in this procedure, $\ell(e') < \ellnew(e_v)$. 
\end{observation}

The proof of~\cref{obs:case2-r2} is exactly the same as~\Cref{Ob:basic} and is thus omitted. 

\subsubsection{Case 3: when $1000 \cdot \alpha^{\ell(e)+2} \leq \Gamma(e) < \alpha^{\ell(e)+3}$ and $e$ is weak-costly} \label{sec:weak-costly-r2}

This  is the main part of the  algorithm and differs considerably from the previous algorithm both technically and conceptually. 
 
We define the set of weak hyperedges incident on $e$ as follows: 
\begin{align*}
	\WW(e) &:= \bigcup_{v \in e} \WW(v), \qquad \qquad 
	w_e := \card{\WW(e)}.
\end{align*}
Note that $w_e = \sum_{v \in e} w_v$ as each hyperedge is owned by exactly one vertex and so $\WW(v)$-sets for different $v \in e$ are disjoint. As $e$ is weak-costly and by the bounds on 
$\Gamma(e)$, we have that 
\[
1000 \cdot \alpha^{\ell(e)+1} \leq w_e \leq \alpha^{\ell(e)+2}.
\]
To continue, we need the following definitions:  
\begin{itemize}
	\item We say that a hyperedge $e' \in \WW(e)$ is a \textbf{conflicting} hyperedge if it has a common vertex with at least $999\alpha^{\ell(e)+1}$ other hyperedges in $\WW(e)$. 
\end{itemize}	
	We use $\CC(e)$ to denote the set of 
	conflicting edges and define $c_e := \card{\CC(e)}$; we define $\CC(v) \subseteq \CC(e)$ and $c_v := \card{\CC(v)}$ for each vertex $v \in e$ analogously. 
\begin{itemize}
	\item Any hyperedge $e' \in \WW(e)$ that is not conflicting is called \textbf{non-conflicting}. 
\end{itemize}
	
We use $\bar{\CC}(e)$ to denote the set of non-conflicting edges and define $\bar{c}_e := \card{\bar{\CC}(e)}$; 
we define $\bar{\CC}(v) \subseteq \bar{\CC}(e)$ and $\bar{c}_v := \card{\bar{\CC}(v)}$ for each vertex $v \in e$ analogously. 

The algorithm in this part consists of two main different case, each with some sub-cases. To begin with, 
we iterate $\WW(e)$ and form the set $\CC(e)$. Then, we consider the following two cases: 

\subsubsection*{Case 3-(a): when $c_e \geq 999 \cdot \alpha^{\ell(e)+1}$.} This is the most interesting  case of the algorithm and where we perform our \emph{hyperedge-centric sampling} as opposed to a vertex-centric one. 

We will directly run $\addrandom(\CC(e))$ in this case to insert a hyperedge $\tilde e$ to $M$ (so $\tilde e$ is sampled not from the neighborhood of a vertex but a hyperedge -- our probabilistic analysis shows that
since this is a conflicting hyperedge, this is still going to work for our purpose).  
Note that since $c_e \geq 999 \cdot \alpha^{\ell(e)+1}$, we have $\ellnew(\tilde e) \geq \ell(e)+1$. Finally, we  run $\detsettle(v)$ for any $v \in e$ that is not handled yet to maintain the invariants.

\begin{claim}\label{clm:case3b-r2}
	The runtime of \case{3-(a)} in $\handleedge(e)$ is $O(\Gamma^{\new}(\tilde e) + \alpha^{\ell(\tilde e)+2})$. 
\end{claim}
\begin{proof}
	By~\Cref{clm:addrandom-r2}, running $\addrandom$ takes  $O(\Gamma^{\new}(\tilde e) + \alpha^{\ell(\tilde e)+2})$ time. Moreover, running $\detsettle$ for all remaining vertices takes $O(\Gamma(e))$ time by~\Cref{clm:case1-r2}. 
	As in \case{3}, $\Gamma(e) < \alpha^{\ell(e)+3}$ and $\ell(\tilde e) \geq \ell(e)+1$, we obtain the desired bound.  
\end{proof}

We record the following observation for this case. 

\begin{observation}\label{obs:case3a-r2}
	$(i)$ For the hyperedge $\tilde e$ inserted to $M$ in \case{3-(a)}, $\ellnew(\tilde e) \geq \ellnew(e)+1$. $(ii)$ hyperedge $\tilde e$ is sampled from a set of $999 \cdot \alpha^{\ell(e)+1}$ hyperedges. $(iii)$ 
	hyperedge $\tilde e$ is incident on at least $998 \cdot \alpha^{\ell(e)+1}$ many of the sampled hyperedges. $(iv)$ For any hyperedge $e'$ deleted from $M$ 
	in case (a) of $\testinsert(\tilde e)$ in this procedure, $\ell(e') < \ellnew(\tilde e)$. 
\end{observation}

The first two and the last properties are straightforward and the third one holds as we are sampling over conflicting hyperedges and by their definition. The

\subsubsection*{Case 3-(b): when $c_e < 999 \cdot \alpha^{\ell(e)+1}$.} 

In this case, $\bar{c}_e \geq \alpha^{\ell(e)+1}$ and we use this to our advantage. The sampling procedure of this step is again vertex-centric but now has a different property than before; it can create hyperedges 
at the \emph{same level} as $e$ instead of strictly one more; we show that this is still fine for the purpose of our analysis. 
The procedure in this part is as follows. 

\noindent	
While $\textbf{(1)}$ $e$ is \emph{still} weak-costly \emph{and} $\textbf{(2)}$ there is a vertex $v \in e$ with $\bar{c}_v > \alpha^{\ell(e)}$:\footnote{Notice that both conditions of the while-loop are true at the beginning of the algorithm.} 
	\begin{itemize}
		\item We use $\addrandom(\bar{\CC}(v))$ which inserts the hyperedge $e_v$ to $M$. 

\item We update the values of $\bar{c}_v$ for the remaining vertices $v \in e$ and repeat the loop. 

\end{itemize}
\noindent
Once we are done with the while-loop, we do as follows. If the while-loop was run \emph{more than once}, we skip the following three steps; otherwise: 
\begin{itemize}
	\item We find the set of $T$ of hyperedges that are \emph{not} incident on $e_v$. Since $e_v$ is a non-conflicting hyperedge, $\card{T} \geq \alpha^{\ell(e)+1}$. 
	\item Since there are at most $r$ vertices incident on each hyperedge, there is a vertex $u \in e$ which is incident on $\alpha^{\ell(e)}$ of these hyperedges in $T$, denoted by the set $T(u)$. 
	\item We run $\addrandom(T(u))$ that inserts a hyperedge $e_u$ to $M$.  
\end{itemize}
For simplicity of notation and brevity, with a slight abuse of notation, from now on we consider the above step as another iteration of the while-loop (the only difference is that in the while-loop, we are sampling $e_v$ from \emph{non-conflicting} hyperedges
incident on $v$ while this time, we are sampling $e_u$ from all weak hyperedges of $u$ including both conflicting and non-conflicting hyperedges). Finally, 
\begin{itemize}
	\item If constraint $\textbf{(1)}$ of the while-loop was violated, i.e., $e$ is now strong-costly, we go to either Case 1 or 2 of $\handleedge(e)$ depending on the value of $\Gamma(e)$ currently. 
	\item If constraint $\textbf{(2)}$ was violated, We now simply run $\detsettle(v)$ for any $v \in e$ that is not handled yet to maintain the invariants. 
\end{itemize}

\begin{observation}\label{obs:case3b-r2}
	Consider any hyperedge $e_v$ inserted to $M$ during this process. In case (a) of $\testinsert(e_v)$, any deleted hyperedge $e'$ from $M$ has $\ell(e') < \ell(e_v)$. 
\end{observation}

We shall note that unlike all previous cases, this is a non-trivial statement simply because hyperedge $e_v$ may have the same level as $e$, i.e., $\ell(e_v) = \ell(e)$ as we are only sampling them from a set of 
$\alpha^{\ell(e)}$ hyperedges.  However, this is true because $e_v$ is a weak hyperedge for $e$; what it means that the endpoints of all other vertices incident on $e_v$ other than $v$ are strictly smaller than $v$, thus their matched
hyperedges can be only at a lower level than $v$, proving the observation. 

We also have the following immediate observation by definition (specifically, that we consider the step right after the while-loop which involves vertex $u$ as yet another iteration of the while-loop). 

\begin{observation}\label{obs:case3b-k2}
	In~\case{3-(b)}, at least two vertices are handled in the while-loop, i.e., $k \geq 2$. 
\end{observation}

\begin{claim}\label{clm:case3a-r2}
	Consider \case{3-(b)} in $\handleedge(e)$ and suppose the while-loop is run for $k$ iterations for vertices $v_1,\ldots,v_k$. Then the runtime of this part \underline{ignoring the recursive calls}, is: 
	\[
		O(\sum_{i=1}^{k} \Gamma^{\new}(e_{v_i}) + \alpha^{\ellnew(e_{v_i})+2}). 
	\]  
\end{claim}
\begin{proof}
	Firstly, note that whenever we handle a non-conflicting hyperedge in the while-loop, we may only remove $999 \cdot \alpha^{\ell(e)+1}$ other hyperedges (conflicting or non-conflicting) from the consideration that are neighbor to this hyperedge. 
	This firstly implies that $k = \Omega( w_e/\alpha^{\ell(e)+1})$. 
	
	At the beginning of the procedure, forming the set $\CC(e)$ takes $O(r \cdot w_e) = O(r \cdot k \cdot \alpha^{\ell(e)+1}) = O(k \cdot \alpha^{\ell(e)+2})$ time which is within the budget of the claim as $\ellnew(v_i) \geq \ell(e)$ for all $v_i$.  
	
	Each iteration of the while-loop, by~\Cref{clm:case2-r2}, takes $O(\Gamma^{\new}(e_{v_i}) + \alpha^{\ell(e_{v_i})+2})$ time, which is again within the budget (the time needed for updating the $\bar{c}_v$-values 
	is dominated by the previous term asymptotically).
	
	If at the end we are in case $(1)$ of while-loop termination, there is nothing else left to do as this only involves a recursive call to $\handleedge$. 
	
	Finally, if at the end we are in the case $(2)$ of while-loop termination, then the total number of non-conflicting hyperedges remained incident on vertices of $e$ that are not handled yet is $O(r \cdot \alpha^{\ell(e)})$. 
	The number of conflicting hyperedges is also $O(\alpha^{\ell(e)+1})$ by the premise of \case{3-(b)}. Also, as $e$ is still weak-costly, the total number of strong hyperedges is $O(r \cdot w_e) = O(r^2 \cdot \alpha^{\ell(e)})$. 
	As such, at this point, $\Gamma(e) = O(\alpha^{\ell(e)+2})$ and so by \Cref{clm:case1-r2}, this part takes $O(\Gamma(e)) = O(\alpha^{\ell(e)}+2)$ which is again 
	within the budget as $k \geq 1$ and $\ell(e) \leq \ell(v_i)$ for all $v_i$. 
\end{proof}

\subsubsection{Case 4: when $\Gamma(e) \geq \alpha^{\ell(e)+3}$ and $e$ is weak-costly} \label{sec:weak-costly-r2}

This is the final case of the procedure and can be handled in a manner almost identical to \case{2}; the only difference is that we now pick a vertex with more weak hyperedges as opposed to strong ones and repeat this in
a while-loop. In particular, the procedure is as follows: 

\noindent	
While there is a vertex $v \in e$ with $w_v \geq \alpha^{\ell(e)+1}$:\footnote{Notice that this condition of the while-loop is true at the beginning of the algorithm.} 
\begin{itemize}
	\item We find $\ellstar := \argmax_{\ell} \set{\tildeo_{v,\ell} \geq \alpha^{\ell}}$. 
	\item We then run $\setlevel(v,\ellstar_v)$ followed by $\addrandom(\OOnew(v))$ which inserts the hyperedge $e_v$ to $M$. 
\end{itemize}
Once the while-loop finishes, we compute $\Gamma(e)$ and run either of \case{1}, \case{2}, or \case{3} accordingly. 
This finalizes the description of the update algorithm in this case. The proof of the following claim is identical to~\Cref{clm:case2-r2}. 

\begin{claim}\label{clm:case4-r2}
	Consider \case{4} in $\handleedge(e)$ and suppose the while-loop is run for $k$ iterations for vertices $v_1,\ldots,v_k$. Then the runtime of this part \underline{ignoring the recursive calls}, is: 
	\[
		O(\sum_{i=1}^{k} \Gamma^{\new}(e_{v_i}) + \alpha^{\ellnew(e_{v_i})+2}). 
	\]  
\end{claim}

The proof of above claim is exactly as before and we do not repeat it again. 

\medskip
\noindent
This concludes the description of our update algorithm.


\section{The Analysis of the $O(r^2)$-Update Time Algorithm} \label{analr2}

Recall the notion of $\epoch$ from the previous algorithm and the following equation that bounds the total cost of update at time $t$: 
\begin{align}
	C(t) = O(r) + \sum_{\epochc(t)} \costc(\epoch(e,t)) + \sum_{\epocht(t)} \costt(\epoch(e,t)). \label{eq:charge-r2}
\end{align}

As before, we give a charging argument that charges the cost of each update to the epochs created in that time step and then use this to bound the runtime of the algorithm.

\subsection{Re-distributing Computation Costs to Epochs}\label{sec:det-r2}

The following is an analogue of~\Cref{lem:charge} with the exception that we are now charging ``less time'' to each epoch, specifically, by a factor $\alpha = \Theta(r)$.  

\begin{lemma}\label{lem:charge-r2}
	The total computation cost $C(t)$ of an update at time $t$ in \Cref{eq:charge-r2} can be charged to the epochs in the RHS so that any level-$\ell$ epoch is charged with $O(\alpha^{\ell+2})$ units of cost. 
\end{lemma}

We prove this lemma in a sequence of claims in this section. In the following, whenever we say ``some computation cost can be charged to $\costc(\epoch(e,t))$ or $\costt(\epoch(e,t))$'' we mean that the cost of computation is $O(\alpha^{\ell+2})$ and there 
is a level-$\ell$ epoch at time $t$ that is either created or terminated, respectively -- we ensure that we will not charge such costs to an epoch more than a constant number of times. This then immediately satisfies the bounds in~\Cref{lem:charge-r2}.

\paragraph{Borrowed charge:} An important change in our analysis compared to the previous algorithm is the notion the \textbf{borrowed charge}. Basically, we also  allow \emph{charging $O(\Gamma(e'))$ computation time} 
simply to calling of $\handleedge(e')$ as a subroutine in $\handleedge(e)$ time $t$ (in different cases of the algorithm). To handle this, we assume each run of $\handleedge(e')$ takes at least $O(\Gamma(e'))$ time and then re-distribute this cost in handling $e'$ to 
the epochs created or terminated by $\handleedge(e')$ or as borrowed cost to $\handleedge(e'')$ which is called inside $\handleedge(e')$ as a subroutine (basically, instead of finding an epoch for $e$ to charge this cost to, we let hyperedge $e'$ handle these). 

\begin{remark}
	As we are analyzing the computation times asymptotically, the reader may worry that these borrowed charges start to add up during recursive calls by constant factors and eventually break the bounds. 
	However, as will be evident (and remarked throughout) our analysis, we will only charge the cost of $\handleedge(e)$ to another $\handleedge(e')$ if $\ell(e') > \ell(e)$ with the exception of \case{3-(b)} wherein 
	we charge the cost of $\handleedge(e)$ to \emph{at least two} $\handleedge(e')$ and $\handleedge(e'')$ with $\ell(e') = \ell(e'') = \ell(e)$. As such, in a series of recursive borrowed charges, the cost 
	of the final charging is within a constant factor of the \emph{entire} series, and thus the analysis stands correct. 
\end{remark}

We analyze~\Cref{lem:charge-r2} in each case of the update algorithm for $\handleedge$ and ignore the costs of hyperedge insertions and hyperedge deletions outside $M$ as they only take $O(r)$ time and hence clearly 
satisfy the bounds in~\Cref{lem:charge-r2}. 

\subsubsection{Case 1 of $\handleedge(e)$} 

The total computation time in this case is $O(\alpha^{\ell(e)+2})$ by~\Cref{clm:case1-r2} for the deleted hyperedge $e$ from $M$, so we can charge this time to the termination of $\epoch(e,t)$ in~\Cref{eq:charge-r2}. 
Also, $e$ may have $O(\Gamma(e))$ borrowed cost which by the guarantee on $\Gamma(e)$ in this case is also $O(\alpha^{\ell(e)+2})$ which can be charged to the termination of $\epoch(e,t)$ in~\Cref{eq:charge-r2}. 

This case does not generate any borrowed cost and can be seen as the base case of reduction for the borrowed costs. 

\subsubsection{Case 2 of $\handleedge(e)$}

The total computation time in this case is $O(\Gamma^{\new}(e_v) + \alpha^{\ell(e_v)+2})$ by~\Cref{clm:case2-r2} for the new sampled edge $e_v$ to be (potentially) inserted to $M$. 

When inserting $e_v$ through $\testinsert(e_v)$ (as a subroutine in $\addrandom$), there are two cases: 
\begin{itemize}
	\item Case (a) which results in $e_v$ joining $M$. In this case, $\Gamma^{\new}(e_v) = O(\alpha^{\ell(e_v)+2})$ by definition. We charge the entire cost of this case then 
	to the creation of $\epoch(e_v,t)$ which is within the budget of~\Cref{eq:charge-r2}. 
	\item Case (b) which results in $e_v$ being removed from $M$. In this case, $\alpha^{\ell(e_v)+2} = O(\Gamma^{\new}(e_v))$. We charge the entire cost of this 
	as a borrowed cost to $\handleedge(e_v)$ called in this case which is within the budget of borrowed costs. Moreover by~\Cref{obs:case2-r2}, we have that $\ellnew(e_v) > \ell(e)$ and thus we are only 
	sending borrowed charge to a higher level hyperedge. 
\end{itemize}
\noindent
Moreover, this case may start with a borrowed cost of $O(\Gamma(e))$. As by~\Cref{obs:case2-r2}, for the edge $e_v$ in this case, we have $\alpha^{\ellnew(e_v)} \leq s_v < \alpha^{\ellnew(e_v)+1}$, 
and that $\Gamma(e) = O(r \cdot s_v)$ in \case{2}, we obtain that $\Gamma(e) = O(\alpha^{\ellnew(e_v)+2})$. Thus, the borrowed cost of the algorithm in this case is at most equal to the computation cost itself and thus can be re-distributed in the same manner.  

The remaining costs are recursive calls and are thus ignored here. 

\subsubsection{Case 3-(a) of $\handleedge(e)$}

This case takes $O(\Gamma^{\new}(\tilde e) + \alpha^{\ellnew(\tilde e)+2})$ time by~\Cref{clm:case3b-r2} and creates a new $\epoch(\tilde e,t)$ and so its cost can be charged to this epoch or become a borrowed charge for $\handleedge(\tilde e)$ accordingly, 
exactly as in the description of \case{2} above (by basically considering the two different cases of $\testinsert(\tilde e)$ accordingly). 

Moreover, the borrowed charge of $O(\Gamma(e))$ in this case is $O(\alpha^{\ell(e)+3})$ by the definition of \case{3}. Moreover, by~\Cref{obs:case3a-r2}, $\ellnew(\tilde e) \geq \ell(e)+1$ and thus 
$\Gamma(e) = O(\alpha^{\ellnew(\tilde e)+2})$. As such, this borrowed charge is asymptotically at most the computation cost in this case and can be re-distribute the same way.

\subsubsection{Case 3-(a) of $\handleedge(e)$}
The total time cost in this case is 
\[
O\paren{\sum_{i=1}^{k} \Gamma^{\new}(e_{v_i}) + \alpha^{\ellnew(e_{v_i})+2}},
\]
 by~\Cref{clm:case3a-r2}, where $v_1,\ldots,v_k$ are the handled vertices in the while-loop. 

The main cost of each handled vertex $v_i$ is through $\testinsert(e_v)$. As before, if $\testinsert$ success in placing $e_v$ in $M$, we will have $\Gamma^{\new}(e_{v_i}) = O(\alpha^{\ellnew(e_{v_i})+2})$ 
and thus we can charge the cost of this vertex to the creation of $\epoch(e_v,t)$; if $e_v$ fails, then we charge the cost of $O(\Gamma^{\new}(e_{v_i}))$ to $\handleedge(e_{v_i})$. 

Moreover, we start this case with $O(\Gamma(e))$ borrowed charge. As shown in the proof of~\Cref{clm:case3a-r2}, if we conclude this case without any recursion to \case{1} or \case{2}, then $\Gamma(e)$
is asymptotically $O(\sum_{i=1}^{k} \Gamma^{\new}(e_{v_i}) + \alpha^{\ellnew(e_{v_i})+2})$ and thus 
we can re-distribute this charge again as a computation cost of this hyperedge. If we recursively go to either of \case{1} or \case{2}, we can continue the analysis in those cases here by the bound $\Gamma^{\new}(e)$ (note that this is not really a recursive call or a borrowed cost -- we simply continue the execution of $\handleedge$ in another case and this can never form a loop as argued before).

An important remark is in order here. Unlike all other cases in the algorithm, here when calling $\handleedge(e_{v_i})$ we are no longer guaranteed that $\ell(e_{v_i}) > \ell(e)$ and indeed 
it might be the case that $\ell(e_{v_i}) = \ell(e)$. Nevertheless, by~\Cref{obs:case3b-k2}, here we will have $k \geq 2$, and thus we can re-distribute the borrowed cost of $\handleedge(e)$ itself between \emph{two} borrowed cost 
at the same level, ensuring that the recursive chain of borrowed costs to not get unboundedly large. 

\subsubsection{Case 4 of $\handleedge(e)$} 

The cost in this case is exactly as in the previous \case{3-(b)} and \case{2} and we do not repeat the same argument. We should only note that here, unlike \case{3-(b)}, any recursive call to $\handleedge$ will result 
in a higher level hyperedge to be considered (and not two hyperedges of potentially the same level).

\medskip

This concludes the proof of~\Cref{lem:charge-r2}. 

\subsection{Natural and Induced Epochs} 

The remainder of our analysis in these parts are straightforward extensions of the $O(r^3)$ update time algorithm and we only focus on highlighting the differences. 
As in our previous algorithm, we are going to rely on~\Cref{assumption1} and then further re-distribute the costs of epochs from induced epochs to natural ones (and then across levels). 

We prove the following lemma as an analogue of~\Cref{lem:induced-natural} for the new algorithm. 

\begin{lemma}\label{lem:induced-natural-r2}
	The total cost charged to all induced epochs can be charged to natural epochs so that any level-$\ell$ natural epoch is charged with the costs of at most $r-1$ induced epochs at \emph{lower levels}.
\end{lemma}
\begin{proof}
	The proof of this is exactly the same as~\Cref{lem:induced-natural} by noting that: 
	\begin{itemize}
		\item In \case{1}: there are no induced epochs; 
		\item In \case{2}: for any induced $\epoch(e',t)$, $\ell(e') < \ell(e)$ by~\Cref{obs:case2-r2};
		\item In \case{3-(a)}: for any induced $\epoch(e',t)$, $\ell(e') < \ell(e)$ by~\Cref{obs:case3a-r2};
		\item In \case{3-(b)}: for any induced $\epoch(e',t)$, $\ell(e') < \ell(e)$ by~\Cref{obs:case3b-r2}; notice the important distinction in this case that even though $\handleedge$ may be called recursively on a level hyperedge, we still have that the induced epochs are one lever lower by focusing on \emph{weak} hyperedges; 
		\item In \case{4}: for any induced $\epoch(e',t)$, $\ell(e') < \ell(e)$ as explained in the algorithm.  
	\end{itemize}
\noindent
	We can then apply the same analysis as in~\Cref{lem:induced-natural} and conclude the proof. 
\end{proof}

The rest of this analysis is then exactly as before by defining the notion of induced and natural levels and the recursive cost of any level-$\ell$ epochs. 
In particular, if we denote the highest recursive cost of any level-$\ell$ epoch by $\hat C_\ell$, for any level $\ell \ge 0$, 
then, by replacing~\Cref{lem:charge-r2,lem:induced-natural-r2} instead of~\Cref{lem:charge,lem:induced-natural}, we obtain that for the new algorithm:  
\[
	\hat C_{\ell} \le   O(\alpha^{\ell+2}) + \sum_{i=1}^{k} \hat 2 \cdot C_{\ell_i},
\]
where $\ell_i < \ell$ for $i \in [k]$ and $k \leq r-1$, with the base case of $C_0 = O(\alpha^2)$. This in turn translates to: 
\[
	\hat C_{\ell} = O(\alpha^{\ell+2}),
\]
which proves the following lemma. 
\begin{lemma} \label{important-r2}
For any $\ell \ge 0$, the recursive cost of any level-$\ell$ epoch is bounded by $O(\alpha^{\ell+2})$.
\end{lemma}

The sum of recursive costs over all fully-natural epochs is equal to the sum of actual costs over all  epochs (fully-natural, semi-natural and induced)
that have been \emph{terminated} throughout the update sequence, which, by~\Cref{assumption1}, bounds the total runtime of the algorithm.

\subsection{The Runtime Analysis}

Let $Y$  be the r.v.\ for the sum of recursive costs over all fully-natural epochs terminated during the entire sequence. As argued in~\Cref{important-r2},
$Y$ will be equal to the total runtime of the algorithm.  
\begin{lemma}  \label{expect-r2}
Let $t$ denote the number of updates in the algorithm (under~\Cref{assumption1}). With high probability, $Y = O(t \cdot \alpha^2 + N \log N \cdot \alpha^2)$. 
\end{lemma}  

We prove this lemma in~\Cref{app:r2}. We note that due to the new hyperedge-centric sampling procedure of our new algorithm, this analysis requires some non-trivial components (in particular, by showing that conflicting hyperedges 
are ``as good as'' vertex-centric sampling). This allows us to conclude the following theorem. 

\begin{theorem}
For any integer $r > 1$, starting from an empty rank-$r$ hypergraph on a fixed set of vertices, a maximal matching can be maintained 
over any sequence of $t$ hyperedge insertions and deletions
in $O(t \cdot r^2)$ time in expectation and with high probability. 
\end{theorem}

This in turn allows us to obtain the following theorem for the dynamic set cover problem, formalizing our main result from the introduction. 

\begin{theorem}
For any integer $f > 1$, starting from an empty set system on a fixed collection of $m$ sets over universe $[n]$, an $f$-approximation of the set system can be maintained 
over any sequence of $t$ element insertions and deletions---wherein each element belongs to at most $f$ sets---in $O(t \cdot f^2)$ time in expectation and with high probability. 
\end{theorem}


 
\section{Details of the Probabilistic Analysis for $O(r^3)$ Update Time}\label{app:r4}

While by~\Cref{inv:level}, all endpoints of a matched hyperedge are of the same level, which also defines the level of the matched hyperedge, this may not be the case for an unmatched hyperedge. By~\Cref{inv:level}(iii), the level of an unmatched hyperedge is the level of its highest endpoint.
For the sake of analysis, we also define the level of a temporarily deleted hyperedge $e'$ to be the level of the unique matched hyperedge $e$ such that $e' \in \mathcal{D}(e)$; recall that hyperedge $e'$ is  deleted from all the data structures temporarily, until hyperedge $e$ is deleted from the matching.
We say that a hyperedge $e$ is \emph{deleted at level} $\ell$ (from the hypergraph) if its level (under the above notation)
is $\ell$ at the time of the deletion.
This means that each hyperedge is deleted at a single level. 

Let $S_\ell$ denote the sequence of hyperedge deletions at level $\ell$,  write $|S_\ell| = t_\ell$, and note that $t_\ell$ is a random variable.
Denoting by $t_{del}$ the total number of deletions, we have $\sum_{\ell \ge 0} t_\ell = t_{del} \le t$.
(The respective upper bound from \cite{Sol16} is $\sum_{\ell \ge 0} t_\ell \le 2t_{del} \le 2t$;
a direct generalization leads to an upper bound of  $\sum_{\ell \ge 0} t_\ell \le r \cdot t$, but under our new definitions each hyperedge is deleted at a single level, which was made possible by the idea of temporarily deleting unmatched hyperedges that we sample over from all the data structures, so that they are not re-sampled again for as long as the corresponding matched hyperedge remains matched, thereby saving a factor of $r$.)
Let $T_\ell = T_\ell(S_\ell)$ be the r.v.\ for the total number of epochs (fully-natural, semi-natural and induced) terminated at level $\ell$.  Since the final hypergraph is empty under~\Cref{assumption1}, any created epoch will get terminated throughout the update sequence.
Consequently, $T_\ell$ designates the total number of epochs created at level $\ell$ and $\sum_{\ell \ge 0} T_\ell$ designates the total number of epochs created
over all levels throughout the entire update sequence.
Also, let $X_\ell = X_\ell(S_\ell)$ be the r.v. for the number of fully-natural epochs terminated at level $\ell$ for the update sequence $S_\ell$,
and let $Y_\ell$ be the r.v.\ for the sum of recursive costs over these epochs.
By \Cref{important}, $Y_{\ell} = O(\alpha^{\ell+3}) \cdot X_\ell$. 
By definition, no epoch of any induced level $j$ is fully-natural, and so $X_j  = Y_j = 0$ for any such level.
Recall that $Y$ is the r.v.\ for the sum of recursive costs over all fully-natural epochs terminated during the entire sequence.
Thus we have $Y = \sum_{\ell \ge 0} Y_\ell$.  

Fix an arbitrary level $\ell, 0 \le \ell \le L$.
Consider any level-$\ell$ epoch initiated by some vertex $u$ at some update step $l$,
and let $\OOnew(u)$ denote the set of $u$'s owned hyperedges at that time.
By Observation \ref{Ob:basic},  $|\OOnew(u)| \ge \alpha^\ell$,
and the matched hyperedge, denoted by $e_u$, is chosen uniformly at random among the hyperedges of $\OOnew(u)$.
All unmatched edges of $\OOnew(u)$ are temporarily deleted until
that epoch is terminated, i.e., when its matched hyperedge $e_u$ is deleted from the matching;
the \emph{(interrupted) duration of the epoch} is defined as the number 
of owned hyperedges $\OOnew(u)$ of $u$ at time $l$ 
that get deleted  from the graph between time step $l$ and the epoch's termination.
(All these hyperedge deletions occur at level $\ell$ by the definition above.)
The \emph{uninterrupted duration} of an epoch is defined as the number of owned hyperedges $\OOnew(u)$ of $u$ at time $l$ that get deleted from the graph between time step $l$ and the time that the random matched hyperedge $e_u$ is deleted \emph{from the graph}.
(By Assumption~\ref{assumption1}, all hyperedges get deleted eventually, hence 
both notions of duration  are well-defined.)
For a natural epoch, its uninterrupted duration equals its duration. However, for an induced epoch,
its uninterrupted duration may be significantly larger than its duration.

\begin{observation} \label{obs:duration}
If there are $q$   level-$\ell$ epochs with (interrupted) durations  $\ge \delta$,
then $q \le t_\ell / \delta$.
\end{observation}
\noindent {\bf Remark.} In the respective observation from \cite{Solom16}, the upper bound is $q \le 2t_\ell / \delta$.
A direct generalization leads to an upper bound of  $q \le r \cdot t_\ell / \delta$, but we apply a more careful argument here.  
\begin{proof}
Any hyperedge $e = (u_1,u_2,\ldots,u_k), k \le r$ is associated with a single epoch, by the design of the algorithm.
Indeed, once an epoch is created, all the unmatched edges that we sampled over for its creation get temporarily deleted from all the data structures until the epoch terminates, hence they cannot be sampled over for the creation of any other epoch, 
while the matched edge is trivially associated only with its own epoch.
Since each of the hyperedge deletions that define the durations of these $q$ level-$\ell$ epochs
is associated with just a single epoch, 
the total number of these deletions 
is at most $t_\ell$, hence $q \le t_\ell / \delta$.
\end{proof}

The uninterrupted durations of distinct level-$\ell$ epochs are not necessarily independent.
However, by  \Cref{Ob:basic},
it deterministically holds that $|\OOnew(u)| \ge \alpha^\ell$.
Let us condition on all the random bits used by the algorithm prior to the ones used to sample the level-$\ell$ matched hyperedge for $u$.
The probability that the uninterrupted duration of the corresponding epoch
is $k$ is $1/|\OOnew(u)| \le 1/\alpha^\ell$, for any $k$.
A level-$\ell$ epoch is called \emph{$\mu$-short} if its \emph{uninterrupted} duration is at most $\mu \cdot \alpha^{\ell}$, for some parameter $0 \le \mu \le 1$.
We derive the following corollary.
\begin{corollary} \label{corborder}
For any $0 \le \mu \le 1$, the probability that an epoch is $\mu$-short is $\le \mu$,
independent of any random bits used by the algorithm other than those for sampling the epoch's matched hyperedge.
\end{corollary}

Write $\eta = 1/16e$; let $T'_\ell$ be the r.v.\ for the number of level-$\ell$ epochs that are $\eta$-short (hereafter, \emph{short}),
and let $T''_\ell = T_\ell - T'_\ell$ be the r.v.\ for the number of remaining level-$\ell$ epochs, i.e., those that are not short.
Let $A_\ell$ be the event that both $T_\ell > 5\log N$ and $T'_\ell \ge T_\ell / 4$  hold. 
\begin{claim} \label{first1}
$\Prob(A_\ell) \le 8/ (3N^5)$.
\end{claim}
\begin{proof}
Fix two parameters $q$ and $j$, with $j \ge q/4$,
and consider any $q$ level-$\ell$ epochs $\cE_1,\ldots,\cE_q$, ordered by their creation time.
We argue that the probability that precisely $j$ particular epochs among these $q$ are short is at most $\eta^j$.
Denote by $B^{(i)}$ the event that $\cE_i$ is short, for $1 \le i \le j$.
Applying induction together with \Cref{corborder}, we get that $\Prob(B^{(i)} ~\vert~ B^{(1)} \cap B^{(2)} \cap \ldots B^{(i-1)}) \le \eta$.
Consequently,
$$\Prob(B^{(1)} \cap B^{(2)} \cap \ldots \cap B^{(j)}) ~=~ \Prob(B^{(1)}) \cdot \Prob(B^{(2)} ~\vert~ B^{(1)}) \cdot \ldots \cdot \Prob(B^{(j)} ~\vert~ B^{(1)} \cap B^{(2)}
\cap \ldots B^{(j-1)}) ~\le~ \eta^j.$$

For any sample space that consists of $q$ arbitrary level-$\ell$ epochs,
there are ${q \choose j}$ possibilities to choose $j$ short epochs among them.   
As we have shown, each such possibility occurs with probability at most $\eta^j$, 
hence $\Prob[T_\ell = q \cap T'_\ell = j] \le {q \choose j} \eta^j$.

Noting that ${q \choose j} \le (eq/j)^j \le (4e)^j$ and recalling that $\eta = 1/16e$, we have ${q \choose j} \eta^j \le (1/4)^j$.
Hence,
\begin{eqnarray*}
\Prob(A_\ell) &=& \Prob[T_\ell  >  5\log N \cap T'_\ell \ge T_\ell / 4]
~=~ \sum_{q > 5 \log N} \sum_{j = q/4}^q \Prob[T_\ell = q \cap T'_\ell = j]
\\ &\le& \sum_{q > 5 \log N} \sum_{j = q/4}^q {q \choose j} \eta^j
~\le~ \sum_{q > 5 \log N} \sum_{j = q/4}^q (1/4)^j
~\le~
\sum_{q > 5 \log N} 4/3((1/2)^q)
\\ &\le& 8/3 ((1/2)^{5\log N}) ~\le~ 8/ (3N^5). \qed
\end{eqnarray*} 

\end{proof}

\begin{claim} \label{negate}
Let $c$ be a sufficiently large constant. If $\neg A_{\ell}$, then $Y_\ell < c\alpha^3(t_\ell + \alpha^\ell \cdot \log N)$.
\end{claim}
\begin{proof}
Recall that $Y_\ell = 0$, for any induced level $\ell$. We   henceforth assume that $\ell$ is a natural level.

If $\neg A_\ell$, then either $T_\ell \le 5\log N$ or $T'_\ell < T_\ell / 4$ must hold.
In the former case $X_\ell \le T_\ell \le 5\log N$, and we have $Y_\ell ~=~ O(\alpha^{\ell+3}) \cdot X_\ell ~\le~ O(\alpha^{\ell+3}) \cdot  5\log N ~<~ c\alpha^3(t_\ell + \alpha^\ell \cdot \log N).$

Next, suppose that $T'_\ell < T_\ell / 4$.
In this case $T''_\ell > 3T_\ell/4$.
Since $\ell$ is a natural level, at least half of the $T_\ell$ epochs are (fully-)natural, i.e., $X_\ell \ge T_\ell/2$. It follows that at least a quarter of the $T_\ell$ epochs that terminated at level $\ell$ are both (fully-)natural and not short.
Denoting  by $X''_\ell$ the r.v.\ for the number of epochs that are both (fully-)natural and not short,
we thus have $X''_\ell \ge T_\ell/4 \ge X_\ell / 4$.
Since these $X''_\ell$ epochs are natural, the duration of each of them is equal to its uninterrupted duration, and thus it exceeds $\eta \alpha^\ell$.
By Observation \ref{obs:duration}, $X''_{\ell} < t_\ell / (\eta \alpha^\ell)$.
We conclude that
$$Y_\ell ~=~ O(\alpha^{\ell+3}) \cdot X_\ell ~\le~ O(\alpha^{\ell+3}) \cdot 4X''_\ell ~\le~ O(\alpha^{\ell+3}) \cdot 4t_\ell / (\eta \alpha^\ell)   ~<~ c\alpha^3(t_\ell + \alpha^\ell \cdot \log N). \qed $$

\end{proof}

Let $A$ be the event that $Y > c(t \cdot \alpha^3 + 2N\log N \cdot \alpha^3)$.
\Cref{negate} yields the following corollary.
\begin{corollary} \label{second2}
If $A$, then $A_0 \cup A_1 \cup \ldots A_{\log_\alpha N}$.
\end{corollary}
\begin{proof}
We assume that $\neg A_0 \cap \neg A_1 \cap \ldots \neg A_{\log_\alpha N}$ holds, and show that $A$ cannot hold.
Indeed, by Claim \ref{negate}, we have that $Y_\ell < c\alpha^3(t_\ell + \alpha^\ell \cdot \log N)$ for each $\ell \ge 0$.
Recall that $\sum_{\ell \ge 0} t_\ell \le t$. Noting that $\sum_{\ell \ge 0} \alpha^\ell \le 2N$, we conclude
 that
\begin{align*} 
Y ~=~ \sum_{\ell \ge 0} Y_\ell ~<~ \sum_{\ell \ge 0} c\alpha^3(t_\ell + \alpha^\ell \cdot \log N)  ~\le~ 
c(t \cdot \alpha^3 + 2N\log N \cdot \alpha^3). \qed
\end{align*}

\end{proof}

\Cref{first1} and~\Cref{second2} imply that
\begin{align*}
 \Prob(A) ~\le~ \Prob(A_0 \cup A_1 \cup \ldots A_{\log_\alpha N}) ~\le~ \sum_{\ell \ge 0} \Prob(A_\ell) ~\le~ (\log_\alpha N + 1) (8/ (3N^5)) ~=~ O(\log N / N^5).
\end{align*}
It follows that $Y$ is upper bounded by $O(t \cdot \alpha^3 + N\log N \cdot \alpha^3)$ with high probability, as required.

\subsection{Expected bound} \label{justify}
We have shown that the runtime of the algorithm is $O(t \cdot \alpha^3 + N \log N \cdot \alpha^3)$ with high probability.
The same bound also holds in expectation, as shown next.

In the next lemma and proof we do not aim for a tight runtime bound (for simplicity).
\begin{lemma} \label{detrem}

Deterministically, it holds that $Y = O(t \cdot N^4)$.
\end{lemma}
\begin{proof}
Clearly, any hyperedge update besides a deletion of a matched hyperedge can be handled trivially within the required time bound.

A deletion of a matched hyperedge may trigger, in addition to some low-cost operations, a long sequence of calls to $\rand(\cdot)$, where the same vertex can be called multiple times in this sequence. 
We define a  \emph{potential function} $f(G) = \sum_{v \in V} {\alpha^{\ell_v}}$ for the dynamic graph $G$ with respect to the dynamic level assignment of its vertices.
Each call to Procedure $\rand$ may trigger at most $r$ calls to the procedure of~\Cref{sec:case-M-delete}, where the latter calls correspond to hyperedges of strictly lower level than the one triggering these calls,
from which it follows that each call to $\rand$ must increase the potential by at least one unit more than the calls to the procedure of~\Cref{sec:case-M-delete} that it triggers may decrease it.
Each such call to the procedure of~\Cref{sec:case-M-delete} may, in turn, trigger more calls to $\rand$,
but the potential growth due to each such call to   $\rand$ is at least 1.

Since the potential value is upper bounded by $|V| \cdot \alpha^{\log_\alpha N} \le |V| \cdot N$ at all times, it follows that the total number of calls to   $\rand$ is at most $|V| \cdot N$. 
Since the runtime of a single call to $\rand$ (disregarding further recursive calls) is naively at most $O(r \cdot |E|)$, the total time spent per a single update step is at most $|V| \cdot N \cdot O(r \cdot |E|) = O(N^4)$, hence $Y = O(t \cdot N^4)$.
\end{proof}

\begin{corollary}
 $\Expect(Y) = O(t \cdot \alpha^3 + N \log N \cdot \alpha^3)$.
\end{corollary}
\begin{proof}
By definition, if event $A$ does not occur, then we have $Y \le c(t \cdot \alpha^3 + 2N\log N \cdot \alpha^3)$.
Event $A$ occurs with probability $O(\log N / N^5)$, and then $Y = O(t \cdot N^4)$ by \Cref{detrem}.
Hence 
\[
\Expect(Y) \le c(t \cdot \alpha^3 + 2N\log N \cdot \alpha^3) +  O(\log N / N^5) \cdot O(t \cdot N^4)
= O(t \cdot \alpha^3 + N \log N \cdot \alpha^3),
\]
concluding the proof. 
\end{proof}
\section{Details of the Probabilistic Analysis for $O(r^2)$ Update Time}\label{app:r2}
We now explain how to adapt and strengthen the probabilistic analysis of \Cref{app:r4} (for update time $O(r^3)$) to achieve the required improvement for update time $O(r^2)$. Despite the significant modifications required in the algorithm itself, the probabilistic analysis is almost identical, with one small yet significant difference; 
next, we describe this difference and provide the changes to the analysis of \Cref{app:r4} required 
for dealing with it.

Fix an arbitrary level $\ell, 0 \le \ell \le L$.
In the algorithm of Section \ref{sec:r4}, any random sampling of a level-$\ell$ matched edge is {\em vertex-centric}, meaning that 
it is initiated by some vertex $u$--- if $\OOnew(u)$ denotes the set of $u$'s owned hyperedges at that time, 
then the matched 
hyperedge, denoted by $e_u$, is chosen uniformly at random among the hyperedges of $\OOnew(u)$, where $|\OOnew(u)| \ge \alpha^\ell$.
In the algorithm of Section \ref{sec:r2}, in addition to vertex-centric random samplings as before, we also have {\em hyperedge-centric} random samplings of matched hyperedges, which are initiated by some hyperedge $e$--- if $\OOnew(e)$ denotes the set of hyperedges owned by the endpoints of $e$ at that time, then the matched 
level-$\ell$ hyperedge, denoted by $\tilde e$, is chosen uniformly at random among a (carefully chosen) subset of $\OOnew(e)$, namely, the set of {\em conflicting-hyperedges} of $e$, denoted by $\CC(e)$. 
Additionally, let $NC(\tilde e)$ denote the set of hyperedges in $\CC(e)$ which are \emph{not} intersecting $\tilde e$
and let $\card{I(\tilde e)}$ denote the complementary set of hyperedges that intersect $\tilde e$,
i.e., $\CC(e) = NC(\tilde e) \cup I(\tilde e)$. We have that 
\begin{enumerate}
\item $\card{\CC(e)} \geq 100\alpha^{\ell}$;
\item $\card{NC(\tilde e)} \leq \alpha^{\ell}$; 
\item $\card{I(\tilde e)} \geq 99\alpha^{\ell}$. 
\end{enumerate}
{\bf Remark.} The leading constants could be scaled down by a factor of 100, to 1, 100, 99/100, respectively, to coincide more naturally with 
the vertex-centric samplings; to this end, we'll of course need to tweak the algorithm accordingly.
However, we chose to work with larger constants in the algorithm intentionally, to emphasize the difference between these two methods of samplings.

All unmatched edges of ${I(\tilde e)}$ 
are temporarily deleted from the time of the sampling, denoted by $l$, until
that epoch is terminated, i.e., until its matched hyperedge $\tilde e$ is deleted from the matching;
as before, we can define the (interrupted) duration with respect to the entire sample space 
${\CC(e)}$ from which we randomly chose $\tilde e$, but note that the edges of $NC(\tilde e)$ are not temporarily deleted as part of the creation of the epoch.
Consequently, it is important to define instead a different notion, with respect to the subset
${I(\tilde e)}$ of  ${\CC(e)}$;
the \emph{intersecting (interrupted) duration of the epoch} is the number 
of hyperedges from ${I(\tilde e)}$
that get deleted  from the graph between time step $l$ and the epoch's termination.
(All these hyperedge deletions occur at level $\ell$ by the same definition as in \Cref{app:r4}.)
The \emph{uninterrupted duration} of an epoch is defined as the number of 
hyperedges from ${\CC(e)}$ at time $l$ that get deleted from the graph between time step $l$ and the time that the random matched hyperedge $\tilde e$ is deleted \emph{from the graph}; 
the \emph{intersecting uninterrupted duration} of an epoch is defined as the uninterrupted duration of the epoch
minus $\alpha^{\ell}$.
Recall that for a natural epoch, its uninterrupted duration equals its duration, which was crucial for analyzing vertex-centric epochs. 
Since $\card{NC(\tilde e)} \leq \alpha^{\ell}$, the intersecting uninterrupted duration
of any natural epoch is no greater than its intersecting duration; this need not hold of course for induced epochs.  

The next observation holds for the same reason that   \Cref{obs:duration} holds. 
\begin{observation} \label{obs:duration2}
If there are $q$   level-$\ell$ epochs with intersecting (interrupted) durations  $\ge \delta$,
then $q \le t_\ell / \delta$.
\end{observation}

The uninterrupted durations of distinct hyperedge-centric level-$\ell$ epochs are not necessarily independent.
However,  
it deterministically holds that  $\card{\CC(e)} \geq 100\alpha^{\ell}$.
Let us condition on all the random bits used by the algorithm prior to the ones used to sample the level-$\ell$ matched hyperedge for $e$.
The probability that the uninterrupted duration of the corresponding epoch
is $k$ is $1/\card{\CC(e)} \le 1/(100\alpha^\ell)$, for any $k$;
thus the probability that the intersecting 
uninterrupted duration of the corresponding epoch
is $k - \alpha^\ell$ is $1/\card{\CC(e)} \le 1/(100\alpha^\ell)$, for any $k$.
A hyperedge-centric level-$\ell$ epoch is called \emph{intersecting $\mu$-short} if its intersecting \emph{uninterrupted} duration is at most $\mu \cdot 100 \alpha^{\ell}$, or equivalently its uninterrupted duration is at most 
$(\mu + 1/100) \cdot 100 \alpha^{\ell}$, for some parameter $0 \le \mu \le 1$.
We derive the following corollary for hyperedge-centric epochs (cf.\ \Cref{corborder} for vertex-centric epochs).
\begin{corollary} \label{corborder2}
For any $0 \le \mu \le 1$, the probability that a hyperedge-centric epoch is intersecting $\mu$-short is $\le \mu + 1/100$,
independent of any random bits used by the algorithm other than those for sampling the epoch's matched hyperedge.
\end{corollary}

Recall that $\eta = 1/16e$; let $T'_\ell$ be the r.v.\ for the number of level-$\ell$ epochs that are either 
(hyperedge-centric) intersecting $(\eta - 1/100)$-short or (vertex-centric) $\eta$-short as in the definition from \Cref{app:r4}; abusing notation a bit, we shall refer to all such epochs, both vertex-centric and hyperedge-centric, as \emph{short}.
Equipped with this augmented definition of short epochs and with the above observation and corollary,
the rest of the analysis of  \Cref{app:r4} carries over smoothly with a few very minor changes. 
Let $T''_\ell = T_\ell - T'_\ell$ be the r.v.\ for the number of remaining level-$\ell$ epochs, i.e., those that are not short.
Let $A_\ell$ be the event that both $T_\ell > 5\log N$ and $T'_\ell \ge T_\ell / 4$  hold. 
The following claim is proved exactly in the same way as in  \Cref{app:r4}. 
\begin{claim} \label{first12}
$\Prob(A_\ell) \le 8/ (3N^5)$.
\end{claim}

The following claim is proved as in  \Cref{app:r4}, except that
we also need to use the fact that
the intersecting durations of natural epochs is at least as large as their intersecting uninterrupted duration. 
Recall also that here $Y_\ell ~=~ O(\alpha^{\ell+2}) \cdot X_\ell$.

\begin{claim} \label{negate22}
Let $c$ be a sufficiently large constant. If $\neg A_{\ell}$, then $Y_\ell < c\alpha^2(t_\ell + \alpha^\ell \cdot \log N)$.
\end{claim}

Let $A$ be the event that $Y > c(t \cdot \alpha^2 + 2N\log N \cdot \alpha^2)$.
\Cref{negate22} yields the following corollary.
\begin{corollary} \label{second2}
If $A$, then $A_0 \cup A_1 \cup \ldots A_{\log_\alpha N}$.
\end{corollary}

\Cref{first12} and~\Cref{second2} imply that
\begin{align*}
 \Prob(A) ~\le~ \Prob(A_0 \cup A_1 \cup \ldots A_{\log_\alpha N}) ~\le~ \sum_{\ell \ge 0} \Prob(A_\ell) ~\le~ (\log_\alpha N + 1) (8/ (3N^5)) ~=~ O(\log N / N^5).
\end{align*}
It follows that $Y$ is upper bounded by $O(t \cdot \alpha^2 + N\log N \cdot \alpha^2)$ with high probability, as required.

Also, the same bound also holds in expectation, following the same argument of \Cref{justify}.
\section{More details on the distributed implementation} \label{appdist}
Consider an arbitrary sequence of (hyper)edge insertions and deletions in a distributed $n$-vertex network that initially contains no edges, and recall that $r$ denotes the rank of the graph.
The update time of the na\"{\i}ve centralized algorithm for maintaining a maximal hypergraph matching, which scans the entire neighborhoods of the at most $r$ endpoints of the updated edge and proceeds in the obvious way, is greater than the sum of the degrees of these endpoints by a factor of $r$, as we spend $O(r)$ time to process each scanned edge.
(This is for the case when a matched edge gets removed from the graph; the other cases can be handled trivially in   $O(r)$ time.)

In a straightforward distributed implementation of this algorithm, every ``edge scan'' $e$ done by the centralized algorithm due to some vertex $u_i$ (an endpoint of the updated edge) is simulated via sending a message from 
$u_i$ along $e$, which arrives at all the other endpoints of $e$ during the round, and then in the next (at most) $r-1$ rounds $u_i$ receives messages that were sent along $e$ by each of the other endpoints, in order to determine whether all endpoints of $e$ are free or not. 
Clearly, each vertex $v$ can obtain complete information about its neighbors in two communication rounds, hence this
straightforward distributed implementation of the na\"{\i}ve  (centralized) maximal matching algorithm 
requires $O(r)$ communication rounds following a single edge update.
Moreover, messages of size $O(\log m)$ (and even $O(1)$) suffice for communicating the relevant information.

Nevertheless, as already mentioned in the introduction, the number of messages sent per update via this na\"{\i}ve algorithm may be a factor of $r^2$ times greater than the maximum degree in the hypergraph, which, in turn, could be $\Omega({n \choose r-1}) = \Omega((\frac{n}{r})^r)$. 
It is immediate that the total number of messages sent in this na\"{\i}ve distributed algorithm is at most linear in the sequential runtime spent by the na\"{\i}ve centralized algorithm, because every basic operation of the centralized algorithm corresponding to a vertex can be done locally for free at the corresponding processor and every basic operation corresponding to an edge can be na\"{\i}vely simulated by an exchange of messages along the ``involved'' endpoints -- if the number of involved endpoints is $r'$, with $r' \le r$, then the sequential runtime is $\Theta(r')$ and the number of messages sent is $O(r')$.

It is readily verified that the exact same properties apply to our algorithm as well. This stems from the fact that our algorithm does not rely on any global data structure or on a centralized agent that coordinates the update procedure.
In particular, all data structures used by our algorithm are either vertex-centric or edge-centric (for individual vertices or edges), and can thus be stored at the corresponding vertices (processors) in the obvious way; manipulating on those data structures can be done locally for free.
We note that our algorithm must use unicast rather than broadcast messages,
which allows each processor to communicate differently with each of its neighbors, and more concretely
to communicate with a subset of its neighbors---otherwise there is no hope to achieve a message complexity of $o(\Delta)$;
the same requirement was used in all previous works that got $o(\Delta)$ amortized message complexity (cf.\ \cite{ParterPS16,PS16,AOSS18,KS18,DBLP:journals/corr/abs-2010-16177}). 
Consequently, our analysis for the basic (centralized) setting shows that the amortized update time is $O(r^2)$,  
which directly translates into an amortized bound of $O(r^2)$ on both the message and round complexities of a straightforward distributed implementation of our centralized algorithm. We summarize this in the following statement.

\begin{theorem}
Starting from an empty distributed network on $n$ fixed vertices, a maximal hypergraph matching 
can be maintained distributively (under the $\mathcal{CONGEST}$ local wakeup model) over any sequence of edge insertions and deletions with an amortized \emph{message and round complexities} of $O(r^2)$ in expectation and with high probability. 
	Here, $r$ denotes the rank of the graph.
\end{theorem}

\clearpage
\bibliographystyle{alpha}
\bibliography{general,randomMMbibfile}

\end{document}